\documentclass[11pt]{article}
\usepackage{setspace}
\usepackage{amsmath,amssymb}
\usepackage{mathrsfs}
\usepackage{amsthm}
\usepackage{algorithm, algpseudocode}
\usepackage{url}
\usepackage{fullpage}
\usepackage{color}
\usepackage{graphicx}

\newtheorem{claim}{Claim}[section]
\newtheorem{lemma}[claim]{Lemma}

\newtheorem{theorem}{Theorem}
\newtheorem{proposition}[claim]{Proposition}
\newtheorem{corollary}[claim]{Corollary}
\newtheorem{definition}[claim]{Definition}

\theoremstyle{definition}
\newtheorem{remark}[claim]{Remark}

\def\<{\langle}
\def\>{\rangle}

\def\eps{{\varepsilon}}

\def\id{{\rm I}}
\def\sT{{\sf T}}
\def\proj{{\cal P}}
\def\qproj{{\cal Q}}

\def\P{{\mathbb P}}
\def\prob{{\mathbb P}}

\def\E{{\mathbb E}} %expectation

\def\one{{\sf 1}}

\def\reals{\mathbb{R}}

\def\bbG{\mathbb{G}}
\def\bbV{\mathbb{V}}

\def\normal{{\sf N}}

\def\cC{{\cal C}}
\def\cS{{\cal S}}

\def\ualpha{\underline{\alpha}}

\def\Tr{{\sf {Tr}}}

\def\range{{\sf {range}}}

\def\de{{\rm d}}

\def\cG{\mathcal{G}}

\def\ind{\mathbb{I}}

\newcommand\norm[1]{\left\lVert{#1}\right\rVert}
\newcommand\abs[1]{\left\lvert{#1}\right\rvert}
\def\bsl{\backslash}

\newcommand\myeqref[1]{{Eq.\,\eqref{#1}}}

\def\sC{{\sf Q}}

\def\mle{{\; \preceq \;}}
\def\mge{{\; \succeq \;}}

\def\bG{{\bf G}}

\def\tN{{H}}
\def\tL{{\widetilde{L}}}
\def\tJ{{\widetilde{J}}}

\def\barW{{ \overline{W}}}

\def\sspan{{\sf span}}

\def\ls{{\;\lesssim\;}}
\def\gs{{\;\gtrsim\;}}
\def\barn{\bar{n}}
\def\nCl{{\binom{[n]}{\le 2}}}
\def\nCe{{\binom{[n]}{2}}}
\def\bG{{\mathbf G}}

\DeclareMathAlphabet{\mathpzc}{OT1}{pzc}{m}{it}

\def\frakL{{\mathfrak{L}}}
  
\def\val{{\sf Val}}
\def\cC{\mathcal{C}}

\graphicspath{{../FIG/}{./FIG/}{./}}

\author{Yash~Deshpande\thanks{Department of Electrical Engineering,
    Stanford University} \;\; and \;\;Andrea~Montanari\thanks{Department of
    Electrical Engineering
and Department of Statistics, Stanford University}}

\title{Improved Sum-of-Squares Lower Bounds for\\ Hidden Clique and
  Hidden Submatrix Problems}

\begin{document}
\maketitle

\begin{abstract}
Given a large data matrix $A\in\reals^{n\times n}$, we consider the problem of determining
whether its entries are i.i.d. with some known marginal distribution
$A_{ij}\sim P_0$,
or instead $A$ contains a principal submatrix $A_{\sC,\sC}$ whose entries
have marginal distribution $A_{ij}\sim P_1\neq P_0$. As a special
case, the hidden (or planted) clique problem requires to find a planted
clique in an otherwise uniformly random graph.
 
Assuming unbounded computational resources, this hypothesis testing
problem is statistically solvable  provided $|\sC|\ge C \log n$
for a suitable constant $C$. However, despite substantial effort, no
polynomial time algorithm is known that succeeds with high probability
when $|\sC| = o(\sqrt{n})$.  Recently Meka  and Wigderson 
\cite{meka2013association}, proposed a method to establish lower bounds
within the Sum of Squares (SOS)
semidefinite hierarchy.  

Here we consider the degree-$4$ SOS relaxation, and study the
construction of \cite{meka2013association} to prove that SOS
fails unless $k\ge C\, n^{1/3}/\log n$. 
An argument presented by Barak
implies that this lower bound cannot be substantially  improved unless  the witness
construction is changed in the proof. Our proof uses the moments method
to bound the spectrum of a certain random association scheme, i.e. 
a symmetric random matrix whose rows and columns are indexed by the
edges of an Erd\"os-Renyi random graph.
\end{abstract}

\tableofcontents

\section{Introduction}

Characterizing the computational complexity of statistical estimation
and statistical learning problems is an outstanding challenge. On one
hand,  a large part of research in this area focuses on the analysis
of specific polynomial-time algorithms, thus establishing upper bounds on the problem
complexity. 
On the other, information-theoretic techniques are used to derive fundamental
limits beyond which no algorithm can solve the statistical problem
under study.
While in some cases algorithmic and information-theoretic bounds
match, in many other examples 
a large gap remains in which the problem is solvable assuming
unbounded resources but simple algorithms fail.
The hidden clique and hidden submatrix problems are prototypical
examples  of this category.

In the hidden submatrix problem, we are given a symmetric data matrix $A\in
\reals^{n\times n}$ and two probability distributions $P_0$ and $P_1$
on the real line, with $\E_{P_0}\{X\} = 0$ and $\E_{P_1}\{X\} = \mu >0$.
We want to distinguish between two hypotheses (we set by convention
$A_{ii}=0$ for all $i\in [n]=\{1,2,\dots,n\}$):
\begin{description}
\item[Hypothesis $H_0$:] The entries of $A$ above the diagonal $(A_{ij})_{i<j}$ are i.i.d. 
  random variables with the same marginal  law $A_{ij}\sim P_0$.
\item[Hypothesis $H_1$:] Given a (hidden) subset $\sC\subseteq [n]$
the entries  $(A_{ij})_{i<j}$ are independent with 
\begin{align}
  A_{ij} \sim \begin{cases}
    P_1 &\text{ if } \{i, j\}\subset \sC,\\
    P_0 &\text{ otherwise.}
  \end{cases}
\end{align}
Further,  $\sC$ is a uniformly random subset conditional on its size,
that is fixed $|\sC|=k$.
\end{description}
Of interest is also the estimation version of this problem, whereby the 
special subset $\sC$ is known to exist, and an algorithm is sought
that identifies $\sC$ with high probability. 

This model encapsulates the basic computational challenges underlying 
a number of problems in which we need to estimate a matrix that is
both sparse and low-rank. Such problems arise across genomics, signal
processing, social network analysis, and machine learning
 \cite{shabalin2009finding,johnstone2009consistency,oymak2012simultaneously}.

The hidden clique (or `planted clique') problem \cite{jerrum1992large}  is a special case of
the above setting, and has attracted considerable interest within
theoretical computer science. Let $\delta_{x}$ denote
the Dirac delta distribution at the point $x\in \reals$. The hidden clique problem
corresponds to the distributions 
\begin{align}
P_1 = \delta_{+1}\, ,\;\;\;\; P_0 =
\frac{1}{2}\,\delta_{+1}+\frac{1}{2}\, \delta_{-1}\, . 
\end{align}
In this case, the  data matrix $A$  can be interpreted as the
adjacency matrix of a graph $G$ over $n$ vertices (whereby $A_{ij}=+1$
encodes presence of edge $\{i, j\}$ in $G$, and $A_{ij} = -1$
its absence). Under hypothesis $H_1$, the set $\sC$ induces a clique in the (otherwise)
random graph $G$.
For the rest of this introduction, we shall focus on the hidden clique
problem, referring to Section \ref{sec:Main} for a formal statement of
our general results.

The largest clique in a uniformly random graph has size
$2\log_2 n + o(\log n)$, with high probability \cite{grimmett1975colouring}.
Thus, allowing for exhaustive search, the hidden clique problem can be solved
when $k \ge (2+\eps)\log_2 n$. 
On the other hand, despite significant efforts \cite{alon1998finding, ames2011nuclear, dekel2011finding,feige2010finding, 
deshpande2014finding}, no polynomial time algorithm
is known to work when $k = o(\sqrt{n})$.
As mentioned above, this is a prototypical case for which a large gap
exists between  performances of well-understood polynomial-time
algorithms, and the ultimate information-theoretic (or statistical)
limits. This remark motivated an ongoing quest for computational lower bounds. 

Finding the maximum clique in a graph is a classical NP-hard problem \cite{karp1972reducibility}.
Even  a very rough
approximation to its size is hard to find
\cite{hastad1996clique,khot2001improved}. In particular, it is hard to
detect the presence of a clique of size $n^{1-\eps}$ in a graph with
$n$ vertices.

Unfortunately, worst-case reductions do not imply computational lower
bounds for distributions of random instances dictated by natural
statistical models. Over the last two years there have been
fascinating advances in crafting careful reductions that preserve the instances distribution
in specific cases
\cite{berthet2013complexity,ma2013computational,chen2014statistical,hajek2014computational,cai2015computational}.
This line of work typically establishes that several detection problems
(sparse PCA, hidden submatrix, hidden community) are at least as hard
as the hidden clique problem with $k = o(\sqrt{n})$. 
This approach has two limitations:
\begin{itemize}
\item[$(i)$] It yields  conditional statements relying on the unproven
  assumption that the hidden clique problem is hard. In absence of any
  `completeness' result, this is a strong assumption that calls for further scrutiny.
\item[$(ii)$] Reductions among instance distributions are somewhat
  fragile with respect changes in the distribution. For instance, it
  is not known whether the hidden submatrix problem with Gaussian
  distributions $P_0=\normal(0,1)$ and
  $P_1= \normal(\mu,1)$ is at least as hard as the hidden
  clique problem, although a superficial look might suggest that they
  are very similar.
\end{itemize}

A complementary line of attack consists in proving unconditional lower
bounds for broad classes of algorithms.
 In an early contribution, Jerrum \cite{jerrum1992large}
 established such a lower bound for a class of Markov Chain Monte Carlo
 methods. Feldman et al. 
\cite{feldman2012statistical} considered a query-based formulation of
the problem  and proved a similar result for `statistical
algorithms.' 
Closer to the present paper is the work of Feige and Krauthgamer
\cite{feige2000finding},
who analyzed the Lov\'asz-Schrijver semidefinite programming (SDP)
hierarchy.
Remarkably, these authors proved that $r$ rounds of this hierarchy
(with complexity $n^{O(r)}$) fail to detect the hidden clique unless
$k\gtrsim\sqrt{n}/2^r$. (Here and below we write $f(n,r,\dots)\gtrsim g(n,r,\dots)$ if there exists a constant
$C$ such that $f(n,r,\dots) \ge C\, g(n,r,\dots)$.)

While this failure of the  Lov\'asz-Schrijver hierarchy provides
insightful evidence towards the hardness of the hidden-clique problem, 
an even stronger indication could be obtained by establishing an
analogous result for the Sum of Squares (SOS) hierarchy
\cite{shor1987class,lasserre2001global,parrilo2003semidefinite}.
This SDP hierarchy unifies most convex relaxations developed for this
and similar problems. Its close connection with the unique games
conjecture has led to the idea that SOS might indeed be an
`optimal' algorithm for a broad class of problems \cite{barak2014sum}.
Finally, many of the low-rank estimation problems mentioned above
include naturally quadratic constraints, that are most naturally
expressed within the SOS hierarchy.

The SOS hierarchy is formulated in terms of polynomial optimization
problems. The level of a relaxation in the hierarchy corresponds to the largest
degree $d$ of any monomial whose value is  explicitly treated as a
decision variable. 
Meka and Wigderson \cite{meka2013association} proposed a construction
of a sequence of  feasible solutions, or witnesses (one for each degree $d$), that can be used to prove
lower bounds for the hidden clique problem within the SOS hierarchy. 
The key technical step consisted in proving that a certain moment
matrix is positive semidefinite: unfortunately this part of their proof
contained a fatal flaw. 

In the present paper we undertake the more modest task of analyzing
the Meka-Wigderson witness for the level  $d=4$ of the SOS hierarchy. 
This is the first level at which the SOS hierarchy differs
substantially from the baseline spectral algorithm of Alon,
Krivelevich and Sudakov \cite{alon1998finding}, or from the
Lov\'asz-Schrijver hierarchy.
We prove that this relaxation fails unless 
\begin{align}
  k \gtrsim \frac{n^{1/3}}{\log n} .\label{eq:FirstResult}
\end{align}
Notice that  the natural guess would be that the SOS hierarchy
fails (for any bounded $d$) whenever $k = o(\sqrt{n})$.  While our
result falls short of establishing this, an argument presented in
\cite{BarakLectureNotes} shows that this is a limitation of the
Meka-Wigderson construction. In other words, by refining our analysis
it is impossible to improve the bound (\ref{eq:FirstResult})  except
--possibly-- by removing the logarithmic factor.

Apart from the lower bound on the hidden clique problem, our analysis
provides two additional sets of results:
\begin{itemize}
\item We apply a similar witness construction to the hidden submatrix
problem with entries distributions $P_0=\normal(0,1)$,
$P_1=\normal(\mu,1)$. 
We define a polynomial-time computable statistical test  that is based
on a degree-$4$ SOS relaxation of a nearly optimal combinatorial test.
We show that this fails unless $k \gtrsim \mu^{-1}n^{1/3}/\log n$. 
\item As mentioned above, the main technical contribution consists in
  proving that a certain random matrix is (with high probability)
  positive semidefinite. Abstractly, the random matrix in question is
  function of an underlying (Erd\"os-Renyi) random graph $G$ over $n$
  vertices. The matrix has rows/columns indexed by subsets of size at
  most $d/2=2$, and elements depending by the subgraphs of $G$ induced by
  those subsets. We shall loosely refer to this type of random matrix
  as to a \emph{random association scheme}.

 In order to  prove that this witness is positive semidefinite, we
 decompose the linear space on which it acts into irreducible
 representation of the group of permutations over $n$ objects. We then
 use the moment method to characterize each submatrix defined by this
 decomposition, and paste together the results to obtain our final
 condition for positivity. 

We believe that both the matrix definition
 and the proof technique are so natural that they are likely to be
 useful in related problems. 
\item As an illustration of the last point, our analysis covers the
  case of  Erd\"os-Renyi graphs with sublinear average degree (namely,
  with average degree of order $n^{1-a}$, $a<1/12$). In particular, it
  is easy to derive sum-of-squares lower bounds for finding cliques in
  such graphs.
\end{itemize}

The rest of the paper is organized as follows. In Section
\ref{sec:Main} we state our main technical result, which concerns the
spectrum of random association schemes. We then show that it implies
lower bounds for the hidden clique and hidden submatrix
problem. Section \ref{sec:Strategy} presents a brief outline of the
proof. 
Finally, Section \ref{sec:Proof} presents the proof of our main
technical result.

While this paper was being written, we became aware through 
\cite{BarakLectureNotes} that --in still unpublished work-- Meka,
Potechin and Wigderson proved that the degree-$d$ SOS relaxation is
unsuccessful unless $k\gtrsim n^{-1/d}$. It would be interesting to
compare the proof techniques.

\section{Main results}
\label{sec:Main}

In this section we present our results. Subsection \ref{sec:Positivity}
introduces a feasible random association scheme that is a slight
generalization of the witness developed in \cite{meka2013association}
(for the degree $d=4$ SOS). We state conditions implying that this
matrix is positive semidefinite with high probability. These
conditions are in fact obtained by specializing a more general result
stated in Proposition  \ref{prop:psddecomp}.
We then derive implications for hidden cliques and hidden submatrices.

\subsection{Positivity of the Meka-Wigderson witness}
\label{sec:Positivity}

We will denote by $\bbG(n, p)$ the undirected Erd\"os-Renyi random
graph model on $n$ vertices, with edge probability $p$.
A graph $G = (V, E)\sim \bbG(n,p)$ has  vertex set $V=[n]\equiv\{1,2,\dots,n\}$, and edges set $E$
defined by letting, for each
$i < j \in [n]$, $\{i, j\}\in E$ independently 
with probability $p$. 

The random association scheme $M=M(G,\ualpha)$ can be thought as a generalization of
the adjacency matrix of $G$, depending on the
graph $G$ and  parameters $\ualpha = (\alpha_1, \alpha_2, \alpha_3, \alpha_4)
\in \reals^{4}$. In order to define the matrix $M$ we first require to
set up some notation. For an integer $r$, we let $\binom{[n]}{r}$ denote the set
of all subsets of $[n]$ of size \emph{exactly} $r$, and $\binom{[n]}{\le r}$ 
denote the set of all subsets of size \emph{at most} $r$. We also let $\emptyset$
denote the empty set.

We shall often identify the collections of subsets of size one, $\binom{[n]}{1} = \{ \{i\}:\, i\in [n]\}$
with $[n]$. Also, we identify $\nCe$ with the set 
of ordered pairs $\{(i, j): i, j\in [n], i<j\}$. If $A = \{i,j\}$ with
$i<j$ we call $i$ ($j$) the head (respectively, tail) of $A$ denoted by $h(A)$ (respectively,
$t(A)$). 

Given the graph $G$ and a set $A\subseteq [n]$, we let $G_A$ denote the subgraph
of $G$ induced by $A$. We define the indicator $\cG_A$ 
\begin{align}
  \cG_{A} &= \begin{cases}
    1 &\text{ if } G_A \text{ is a clique},\\
    0 &\text{ otherwise.}
  \end{cases}
\end{align}
For convenience of notation we let 
$\cG_{ij} \equiv \cG_{\{i, j\}}$ and $g_A = \cG_{A} - \E\{\cG_A\} $ be the centered
versions of the variables $\cG_{ij}$.  We also set
$g_{ii} \equiv 0$. 

We can now define the matrix $M = M(G, \ualpha) \in \reals^{\nCl\times\nCl}$ as
follows. For any pair of sets $A, B \in \nCl$ we have:
\begin{align}
  M_{A, B} &= \alpha_{\abs{A\cup B}} \cG_{A\cup B}\, ,
\end{align}
with $\alpha_0  \equiv 1$. 
\begin{theorem} \label{thm:main}
  Suppose $\ualpha$, $p$ satisfy:
  \begin{align}
    \alpha_1 = \kappa \,,\;\;\;\;
    \alpha_2 = 2\frac{\kappa^2}{p} \, ,\;\;\;\;
    \alpha_3 = \frac{\kappa^3}{p^3}\, ,\;\;\;\;
    \alpha_4 = 8\frac{\kappa^4}{p^6} \, ,\;\;\;\;
    p \ge c (\kappa\log n)^{1/4} n^{1/6},
  \end{align}
  for some  $\kappa \in [\log n/n,  n^{-2/3}/\log n]$ and $c$ a large enough absolute constant.  
  If $G \sim \bbG(n, p)$ is a random graph with edge probability $p$ then, for every $n$ large enough, 
\begin{align}
\prob\big\{M(G, \ualpha)\succeq 0\big\}\ge 1 -\frac{1}{n}\, .
\end{align}
\end{theorem}
The proof of this theorem can be found in Section \ref{sec:Proof}. As
mentioned above, a more general set of conditions that imply $M(G,
\ualpha)\succeq 0$ with high probability is given in Proposition
\ref{prop:psddecomp}.
The proof of Theorem \ref{thm:main} consists in checking that the
conditions of Proposition \ref{prop:psddecomp} hold and deriving the consequences.

\subsection{A Sum of Squares lower bound for Hidden Clique}

We denote by $\bbG(n,p,k)$ hidden clique model, i.e. the distribution
over graphs $G=(V,E)$, with vertex set $V=[n]$, a subset $\sC\subseteq
[n]$ of $k$ uniformly random vertices forming a clique, and every other edge present
independently with probability $p$.

The SOS relaxation of degree $d=4$ for the maximum clique problem 
\cite{tulsiani2009csp,BarakLectureNotes} is a semidefinite program,
whose decision variable is a matrix $X\in \reals^{\nCl\times\nCl}$:
\begin{align}
    \text{ maximize }  &\sum_{i\in [n]} X_{\{i\},\{i\}}\, , 
    \label{eq:sosclique}\\
\text{ subject to: } & X\succeq 0,\;\;\;\;X_{S_1,S_2} \in [0, 1]\, ,\nonumber\\
  & X_{S_1,S_2}= 0 \quad\text{ when }
    S_1\cup S_2\text{ is not a clique in } G \, ,\nonumber \\
    &X_{S_1,S_2}= X_{S_3,S_4}\quad \text{for all }
    S_{1}\cup S_{2} = S_{3}\cup S_{4} \, ,\nonumber \\
    &X_{\emptyset,\emptyset} = 1. \nonumber
  \end{align}
Denote by $\val(G;d=4)$ the value of this optimization problem for
graph $G$ (which is obviously an upper bound on the size of the maximum
clique in $G$). We can then try to detect the clique (i.e. distinguish
hypothesis $H_1$ and $H_0$ defined in the introduction), by using the
test statistics
\begin{align}
T(G) = \begin{cases}
0 & \mbox{ if $\val(G;4)\le c_* k$,}\\
1 & \mbox{ if $\val(G;4)> c_* k$.}
\end{cases}\label{eq:CliqueTestSOS}
\end{align}
with $c_*$ a numerical constant. The rationale for this test is as
follows: if we replace $\val(G;4)$ by the size of the largest clique,
then the above test is essentially optimal, i.e. detects the clique
with high probability as soon as $k\gtrsim \log n$ (with $c_*=1$). 

We then have the following immediate consequence of Theorem \ref{thm:main}.
\begin{corollary}\label{cor:clique}
  Suppose $G\sim \bbG(n, 1/2)$. Then, with probability at least
  $1-n^{-1}$, the degree-$4$ SOS relaxation has value
\begin{align}
  \val(G;4)\gtrsim \frac{n^{1/3}}{\log n}\, .
\end{align}
\end{corollary}
\begin{proof}
  Consider $M(\ualpha, G)$ from Theorem \ref{thm:main} (with
  $p = 1/2$). For $M(\ualpha, G)$ to be positive semidefinite with high
  probability, we  set $\kappa = c_0\, n^{-2/3}/\log n$ for some absolute
  constant $c_0$. It is easy to check that $M(\ualpha,G)$ is a
  feasible point for the optimization problem  (\ref{eq:sosclique}).
Recalling that $M_{\{i\},\{i\}} = \alpha_1=\kappa$, we conclude that 
the  objective function at this point is  $n\kappa = c_0n^{1/3}/\log
n$, and the claim follows. 
\end{proof}
We are now in position to derive a formal lower bound on the test (\ref{eq:CliqueTestSOS}).
\begin{theorem}\label{thm:clique}
  The degree-$4$  Sum-of-Squares test for the maximum clique problem,
  defined in Eq.~(\ref{eq:CliqueTestSOS}),
  fails to distinguish between $G\sim \bbG(n, k, 1/2)$ and $G\sim \bbG(n,
  1/2)$ with high probability  if $k \lesssim  n^{1/3}/\log n$. 

In particular, $T(G)=1$ with high probability both for $G\sim \bbG(n, k,1/2)$, and for $G\sim \bbG(n, 1/2)$.
\end{theorem}
\begin{proof}[Proof of Theorem \ref{thm:clique}]
Assume $k \le c_1 n^{1/3}/\log n$ for $c_1$ a sufficiently small constant. 
For $G\sim \bbG(n, 1/2)$, Corollary \ref{cor:clique} immediately
implies that $\val(G;4)\ge c_*k$, with high probability.

For $G\sim \bbG(n, k , 1/2)$, we obviously have $\val(G;4)\ge k$
(because SOS gives a relaxation). To obtain a larger lower bound,
recall that $\sC\subseteq [n]$ indicates the vertices in the clique. 
  The subgraph $G_{\sC^c}$ induced by the set of vertices $\sC^c=[n]\bsl\sC$
  is distributed as $\bbG(n-k, 1/2)$. Further, we obviously have
\begin{align}
\val(G;4)\ge \val(G_{\sC^c};4)\,.
\end{align}
Indeed we can always set to $0$ variables indexed by sets $A\subseteq
[n]$ with $A\not\subseteq \sC^c$.  
Hence,  applying again Corollary \ref{cor:clique}, we deduce that, with probability
  $1-(n-k)^{-1}$, $\val(G;4) \ge C(n-k)^{1/3}/\log (n-k)$, which is
  larger than $c_*k$. Hence $T(G)=1$ with high probability.
\end{proof}

\subsection{A Sum of Squares lower bound for Hidden Submatrix }

As mentioned in the introduction, in the hidden submatrix problem we 
are given  a matrix $A\in\reals^{n\times n}$, which is generated
according with either hypothesis $H_0$ or hypothesis $H_1$ defined there. 
To avoid unnecessary technical complications, we shall consider distributions 
$P_0 = \normal(0, 1)$ (for all the entries in $A$ under $H_0$) and
$P_1 = \normal(\mu, 1)$
(for the entries $A_{ij}$, $i,j\in\sC$ under $H_1$) . 

In order to motivate our definition of an  SOS-based statistical test,  
we begin by introducing a nearly-optimal  combinatorial test, call
it $T_{\rm comb}$. This test essentially look for a principal
submatrix of $A$ of dimension $k$, with average value larger than
$\mu/2$. Formally
\begin{align}
T_{\rm comb}(A) \equiv \begin{cases}
1 & \mbox{ if $\exists x\in\{0,1\}^n$ such that $\sum_{i\in [n]}x_i\le
  k$, and }\\
& \mbox{\;\;\;\; and $\sum_{i,j\in[n],i<j}  A_{ij}x_ix_j\ge \frac{1}{2}\binom{k}{2}\mu$,}\\
0& \mbox{ otherwise.}
\end{cases}
\end{align}
A straightforward union-bound calculation shows that $T_{\rm comb}(\,\cdot\,)$
succeeds with high probability provided $k\gtrsim
\mu^{-2}\log n$.

As in the previous section, the degree-$4$ SOS relaxation of the set
of binary vectors $x\in \{0,1\}^n$
consists in the following convex set of matrices 
\begin{align}
\cC_4(n)\equiv\Big\{X\in \reals^{\nCl\times \nCl}:\;\;\;&\, X\succeq
                                                          0,\;\;\;\;X_{S_1,S_2}
                                                          \in [0, 1]\,,\;\;\;
                                                          \; X_{\emptyset,\emptyset} = 1,\nonumber\\
    &X_{S_1,S_2}= X_{S_3,S_4}\quad \text{for all }
    S_{1}\cup S_{2} = S_{3}\cup S_{4} \Big\}\, .
\end{align}
This suggests the following relaxation of the test $T_{\rm
  comb}(\,\cdot\,)$:
\begin{align}
T(A) = \begin{cases}
1 & \mbox{ if there exists $X\in \cC_4(n)$ such that
 $\sum_{i\in[n]}X_{\{i\},\{i\}}\le k$, and}\\
& \mbox{\;\;\;\;\;$\sum_{i,j\in[n],i<j}  A_{ij}X_{\{i\},\{j\}} \ge c_*\mu k^2$,}\\
0 & \mbox{otherwise}\, .
\end{cases}\label{eq:SubMatrixTest}
\end{align}

We begin by stating a corollary of Theorem \ref{thm:main}.
\begin{corollary}\label{cor:hiddensub}
  Assume $A$ is distributed according to hypothesis  $H_0$,
  i.e. $A_{ij} \sim \normal(0, 1)$ for all $i, j\in [n]$. 
Then, with probability at least  $1-2n^{-1}$, there exists $X\in \cC_4(n)$ such that
\begin{align}
 \sum_{i\in[n]}X_{\{i\},\{i\}}\lesssim \frac{ n^{1/3}}{\log n}\, ,\;\;\;\;  \sum_{i,j\in[n],i<j}
  A_{ij}X_{\{i\},\{j\}} \gtrsim \frac{ n^{2/3}}{(\log n)^2}\, .\label{eq:ValueX}
\end{align}
\end{corollary}
\begin{proof}
Fix $\lambda$ a sufficiently large constant and let $G$
  be  graph with adjacency matrix $\cG$ given by $\cG_{ij} =
  \ind(A_{ij} \ge \lambda)$. Note that this is an Erd\"os-Renyi random
  graph $G\sim \bbG(n,p)$ with edge probability $p = \Phi(-\lambda)$. 
(Throughout  this proof, we let $\phi(z) \equiv e^{-z^2/2}/\sqrt{2\pi}$
denote the Gaussian density,  and  $\Phi(z)\equiv
\int_{-\infty}^z\phi(t) \, \de t$ the Gaussian distribution function.)

 We choose $X = M(G, \ualpha)$ a random association scheme, where $\ualpha$ is set according
  to Theorem \ref{thm:main}, with 
\begin{align}
\kappa = \frac{c_2}{n^{2/3}\log n}\, , 
\end{align}
with $c$ a suitably small constant.
This ensures that  the conditions
  of Theorem \ref{thm:main} are satisfied, whence $X\in\cC_4(n)$ with
  high probability.
Further, by definition
\begin{align}
\sum_{i\in[n]}X_{\{i\},\{i\}} = n\kappa = \frac{c_2\, n^{1/3}}{\log n}\, .
\end{align}

 It remains to check that the second inequality in
  (\ref{eq:ValueX}) hold. We have
\begin{align}
 \sum_{i,j\in[n],i<j}  A_{ij}X_{\{i\},\{j\}}  =
    \frac{2\kappa^2}{p}\sum_{i,j\in[n], i<j} A_{ij} \cG_{ij} \, .
\end{align}
Note that
\begin{align}
\E\Big\{\sum_{i,j\in[n], i<j} A_{ij} \cG_{ij}\Big\} =
  \binom{n}{2}\, \E\big\{A_{12}\,\ind(A_{12}\ge \lambda)\big\} =
  \binom{n}{2} \, \phi(\lambda)
\, .
\end{align}
Note that the random variables $(A_{ij} \cG_{ij})_{i<j}$ are
independent and subgaussian.  By a standard concentration-of-measure
argument we have, with probability at least $1-n^{-2}$, for a suitably
small constant $c'$,
$\sum_{ i<j} A_{ij} \cG_{ij}\ge c' n^2\phi(\lambda)$ and hence
\begin{align}
\sum_{i,j\in[n],i<j}  A_{ij}X_{\{i\},\{j\}}  \gtrsim \kappa^2
  n^2\gtrsim  \frac{n^{2/3}}{(\log n)^2}\, .
\end{align}
\end{proof}

\begin{theorem}
  Consider the Hidden Submatrix problem with entries' distributions
 $P_0 = \normal(0, 1)$, and $P_1 = \normal(\mu, 1)$.

  Then, the degree-$4$  Sum-of-Squares test 
  defined in Eq.~(\ref{eq:SubMatrixTest}),
  fails to distinguish between hypotheses $H_0$ and $H_1$  if $k \lesssim \mu^{-1}n^{1/3}/\log n$. 
In particular, $T(A)=1$ with high probability both under $H_0$ and
under $H_1$. 
\end{theorem}
\begin{proof}
First consider $A$ distributed according to hypothesis $H_0$.
Note that, if $X_0\in\cC_4(n)$ and $s\in [0,1]$ is a scaling factor,
then $s\, X_0\in\cC_4$. Therefore (by choosing $s = ckn^{-1/3}\log n$
for a suitable constant $c$) Corollary \ref{cor:hiddensub} implies
that  with high probability there exists $X\in \cC_4(n)$ such that
\begin{align}
 \sum_{i\in[n]}X_{\{i\},\{i\}}\le k\, ,\;\;\;\;  \sum_{i,j\in[n],i<j}
  A_{ij}X_{\{i\},\{j\}}\gtrsim \frac{ k\, n^{1/3}}{\log n}\, .
\end{align}
Therefore, for  $\mu\, k\le c\, n^{1/3}/\log n$ with $c$ a
sufficiently small constant, we have $\sum_{i<j}
A_{ij}X_{\{i\},\{j\}}\ge c_*\mu\, k^2$ and therefore $T(A) =1$ with
high probability.

Consider next $A$ distributed according to hypothesis $H_1$.
Note that $A = \mu\, \one_{\sC}\one_{\sC}^{\sT} + \widetilde{A}$,
where $\one_{\sC}$ is the indicator vector of set $\sC$, and
$\widetilde{A}$ is distributed according
to $H_0$. Since $\sum_{i<j}A_{ij}X_{\{i\},\{j\}}$ is increasing in $A$, we also have that
$T(\widetilde{A}) = 1$ implies $T(A)=1$. 
As shown above, for $\mu\, k\le c\, n^{1/3}/\log n$, we have
$T(\widetilde{A}) =1$ with high probability, and hence $T(A)=1$.
\end{proof}

\section{Further definitions and proof strategy}
\label{sec:Strategy}

In order to prove $M(G,\ualpha)\succeq 0$, we will actually study a new matrix $N(G, \ualpha) \in \reals^{\nCl
\times\nCl}$ defined as follows:
\begin{align} \label{eq:tMdef}
  N_{A, B} &= \alpha_{\abs{A\cup B}} \prod_{i\in A\bsl B, j\in B\bsl
             A} \cG_{ij}.
\end{align}
Notice that $M_{A, B} = N_{A, B}\cG_A\cG_B$, i.e. $M$ is obtained from
$N$ by zeroing  columns (rows) indexed by sets $A, B$ that do not 
 induce cliques in $G$. Thus,  $N\mge 0$ implies $M\mge 0$. 

We also define the matrix $\tN \in \reals^{\left(\binom{[n]}{1} \cup
    \nCe \right)\times \left( \binom{[n]}{1} \cup \nCe\right)}$ 
that is the Schur complement of $N$ with respect to entry
$N_{\emptyset,\emptyset} = 1$. Formally:
\begin{align}
  \tN_{A, B} &=  N_{A, B} - \alpha_{\abs{A}}\alpha_{\abs{B}}\, ,
\end{align}
where, as before, we define $\alpha_0 = 1$.
Furthermore we denote by $\tN_{a,b}$, for $a , b \in \{1, 2\}$, the
restriction of $\tN$ to rows indexed by $\binom{[n]}{a}$ and columns
indexed by $\binom{[n]}{b}$.  (This abuse of notation will not be a
source of confusion in what follows, since we will always use
explicit values in $\{1,2\}$ for the subscripts $a,b$. )

Since $\tN$ is the Schur complement of $N$, $\tN\succeq 0$ implies
$N\succeq 0$ and hence $M\succeq 0$. The next section is devoted to
prove $\tN\succeq 0$: here we sketch the main ingredients.

Technically, we control the spectrum of $\tN$ by first 
computing eigenvalues and eigenspaces of its expectation $\E\tN$
and then controlling the random part $\tN-\E\tN$ by the moment method,
i.e. computing moments of the form
$\E\Tr\{(\tN-\E\tN)^{2m}\}$. The key challenge is that the simple
triangular inequality 
$\lambda_{\rm min}(\tN)\succeq \lambda_{\rm min}(\E\tN)-
\|\tN-\E\tN\|_2$ is too weak for proving the desired result. We
instead decompose  $\tN$ in its blocks $\tN_{1,1}$, $\tN_{1,2}$,
$\tN_{2,2}$ and prove the inequalities stated in Proposition
\ref{prop:psddecomp},  cf. Eqs.~(\ref{eq:11PSD}) to
(\ref{eq:22PSD}). Briefly, these allow us to conclude that:
\begin{align}
  \tN_{1, 1} &\mge 0. \\
  \tN_{2, 2} &\mge \tN_{1, 2}^\sT \tN_{1, 1}^{-1}\tN_{1,2}, \label{eq:22schur}
\end{align}
which are the Schur complement conditions guaranteeing $\tN\mge 0$. 
While characterizing $\tN_{1,1}$  is relatively easy (indeed this
block is
essentially the adjacency matrix of $G$), the most challenging part of the proof
consists in showing a sufficient condition for \myeqref{eq:22schur}
(see \myeqref{eq:22PSD} below). In order to prove this
bound, we need to decompose $\tN_{2,2}$ and $\tN_{1,2}$ along the
eigenspaces of $\E \tN_{2,2}$, and carefully control each of the
corresponding sub-blocks.

In the rest of this section we demonstrate the essentials of our
strategy to show the weaker assertion $\tN_{2, 2} \mge 0$. 
We will assume that $p$ is order one, for concreteness 
$p = 1/2$ which corresponds to the hidden clique problem. 
It suffices to show that
\begin{align}
\E\tN_{2, 2} \mge \E\tN_{2, 2} -\tN_{2, 2}. \label{eq:22PSDweak}
\end{align}
The expected value $\E\tN_{2, 2}$ has 3 distinct eigenspaces
$\bbV_0, \bbV_1, \bbV_2$ that form an orthogonal decomposition of
$\reals^{\nCe}$. Crucially, these spaces admit a simple
description as follows:
\begin{align}
  \bbV_0 &\equiv \{v \in \reals^{\nCe}: \exists u\in\reals \text{ s.t. }
            v_{\{i, j\}} = u \mbox{ for all } i<j\} \, ,\\
  \bbV_1 &\equiv \{v \in \reals^{\nCe}: \exists u\in\reals^{n},
           \text{ s.t. } \<\one_n, u\> = 0 \text{ and } 
\quad v_{\{i, j\}} = u_i + u_j \mbox{ for all } i<j\}\, ,\\
  \bbV_{2} &\equiv (\bbV_{0}\oplus\bbV_{1})^{\perp}.
\end{align}
If $\proj_a$ is the orthogonal projector onto $\bbV_a$ we have that
$\E\tN_{2, 2} = \lambda_0 \proj_0 + \lambda_1\proj_1 + \lambda_2\proj_2$
where $\lambda_0 \approx n^2\kappa^4, \lambda_1 \approx n\kappa^3$ and
$\lambda_2\approx \kappa^2$ (see Proposition \ref{prop:EM22eig} for a 
formal statement). 

Now, consider the entry indexed by $\{i, j\}, \{k, \ell\}\in \nCe$:
\begin{align}
  (\tN_{2, 2})_{\{i, j\}, \{k, \ell\}} &= -\alpha_2^2 + \alpha_4
  \cG_{ik}\cG_{i\ell}\cG_{jk}\cG_{j\ell} \\
  &= -\alpha_2^2 + \alpha_4 (p+g_{ik})(p + g_{i\ell})(p + g_{jk})(p+g_{j\ell})\\
  &= -\alpha_2^2 + \alpha_4 p^4 + \alpha_4 p^3 (g_{ik} + g_{i\ell} +
  g_{jk} + g_{j\ell}) \nonumber\\ &\quad+ \alpha_4 p^2 (g_{ik}g_{i\ell} + g_{ik}g_{jk}
  g_{jk}g_{j\ell} + g_{i\ell}g_{j\ell} + g_{ik}g_{j\ell} + g_{i\ell}g_{jk})
  \nonumber\\
 &\quad + \alpha_4 p (g_{ik}g_{i\ell}g_{jk} + g_{ik}g_{jk}g_{j\ell} +
  g_{ik}g_{i\ell}g_{j\ell} + g_{i\ell}g_{jk}g_{j\ell}) + \alpha_4
  g_{ij}g_{i\ell}g_{jk}g_{j\ell}. \label{eq:N22approxdecomp}
\end{align}
The decomposition \myeqref{eq:N22approxdecomp} holds only when $\{i, j\}$ and $\{k, \ell\}$ are disjoint. 
Since
the number of pairs $\{i, j\}, \{k, \ell\}$ that intersect are at most
$n^3 \ll n^4$, it is natural to conjecture that these pairs are
negligible, and in this outline we shall indeed assume that this is
true (the complete proof deals with these pairs as well). The random portion
$\E\tN_{2, 2} - \tN_{2, 2}$ involves the last 15 terms of the above
decomposition. Each term is indexed by a pair $(\eta, \nu)$ where $1\le \eta\le
4$ denotes the number of $g_{ij}$ variables in the term and $1\le \nu \le
\binom{4}{\eta}$ the exact choice of $\eta$ (out of 4) variables used. In 
accordance with notation used in the proof, we let $\tJ_{\eta, \nu}$ denote
the matrix with $\{i, j\}, \{k ,\ell\}$ entry is the $(\eta, \nu)$ entry in the
decomposition \myeqref{eq:N22approxdecomp}. See Table \ref{tab:ribbondef} and
\myeqref{eq:tJdef} for a formal definition of the matrices $\tJ_{\eta, \nu}$. 
Hence we obtain (the $\approx$ below is due to the  intersecting pairs, which we have ignored):
\begin{align}
\tN_{2, 2} - \E\tN_{2, 2} \approx \sum_{\eta\le 4}\sum_{\nu\le
    \binom{4}{\eta}} \tJ_{\eta, \nu}\, .
\end{align}
We are therefore left with the task of proving
\begin{align}
   \E\tN_{2, 2} &\mge Q \equiv
    -\sum_{\eta}\sum_{\nu} \tJ_{\eta, \nu}. \label{eq:22PSDweak2}
\end{align}
Viewed in the decomposition given by $\bbV_0, \bbV_1, \bbV_2$, \myeqref{eq:22PSDweak2}
is satisfied if:
\begin{align}
  \begin{pmatrix}
    \lambda_0 &0 & 0\\
    0 &\lambda_1 &0 \\
    0 &0 &\lambda_2
  \end{pmatrix} \mge
  \begin{pmatrix}
    \norm{\proj_0 Q \proj_0}_2 & \norm{\proj_0 Q \proj_1}_2& \norm{\proj_0 Q \proj_2}_2 \\ 
    \norm{\proj_1 Q \proj_0}_2 & \norm{\proj_1 Q \proj_1}_2& \norm{\proj_1 Q \proj_2}_2\\
    \norm{\proj_2 Q \proj_0}_2 & \norm{\proj_2 Q \proj_1}_2& \norm{\proj_2 Q \proj_2}_2\\
  \end{pmatrix} \label{eq:22PSDproj}
\end{align}
The  bulk of the proof is devoted to developing operator norm bounds for the
matrices $\proj_a \tJ_{\eta, \nu}\proj_b$ that hold   with high
probability. We then bound $\proj_a Q\proj_b$ using triangle inequality
\begin{align}
  \norm{\proj_a Q\proj_b}_2 &\le \sum_{\eta, \nu} \big\|\proj_a\tJ_{\eta, \nu}\proj_b\big\|_2.\label{eq:triangle}
\end{align}

The matrices $\tJ_{4, 1}, \tJ_{3, \nu}, \tJ_{2, 1}, \tJ_{2, 6}$ turn out to have an approximate
``Wigner''-like behavior, in the following sense.  
Note that these are symmetric matrices of size $\binom{n}{2}\approx n^2/2$ 
with random zero-mean entries bounded by $\alpha_4$. If their
entries were \emph{independent}, they would have operator norms of order 
$\alpha_4 \sqrt{n^2/2} \approx \kappa^4 n$ \cite{furedi1981eigenvalues}.
Although the entries are actually not independent, the conclusion still holds 
for $\tJ_{4, 1}, \tJ_{3, \nu}, \tJ_{2, 1}, \tJ_{2, 6}$ and they have operator norms
of order $\kappa^4 n$. Hence 
$\|\proj_a \tJ_{\eta, \nu}\proj_b\|_2 \le \|\tJ_{\eta, \nu}\|_2 \approx \kappa^4 n$
for these cases. 

We are now left with the cases $(\tJ_{1, \nu})_{1\le \nu\le 4}$ and $(\tJ_{2, \nu})_{2\le \nu\le 5}$. These
require more care, since their typical norms are significantly larger
than $n$. For instance consider $\tJ_{1, \nu}$ where
\begin{align}
  (\tJ_{1, \nu})_{\{i, j\}, \{k, \ell\}} &= g_{ik}.
\end{align}
Viewed as a matrix in $\reals^{n^2\times n^2}$, $\tJ_{1, \nu}$ corresponds to 
the matrix $\alpha_4 g\otimes (\one_n\one_n)^\sT$ where $\otimes$ denotes the standard
Kronecker product and $g \in \reals^{n\times n}$ is the matrix with $(i, j)$ entry
being $g_{ij}$. By standard results on Wigner random matrices \cite{furedi1981eigenvalues},  $\norm{g}_2 \lesssim \sqrt{n}$ with
high probability. Hence:
\begin{align}
  \norm{g \otimes \one_n\one_n^\sT}_2 = \norm{g}_2
  \norm{\one_n\one_n^\sT}_2 \lesssim n^{3/2},
\end{align}
with high probability. This suggests that $\|\tJ_{1, \nu}\|_2 \lesssim \alpha_4 n^{3/2} \approx \kappa^4 n^{3/2}$ with
high probability. This turns out to be the correct order for all the matrices
$\tJ_{1, \nu}$ and $\tJ_{2, \nu}$ under consideration. 

This heuristic calculation shows the need to be careful with these terms. Indeed, a 
naive application 
of this results yields that $\norm{\proj_a Q\proj_b}_2 \ls \kappa^4 n^{3/2}$. Recalling
\myeqref{eq:22PSDproj}, this imposes that $\lambda_2 \gg \kappa^4 n^{3/2}$. Since 
we have $\lambda_2 \approx \kappa^2$, we obtain the condition $\kappa \ll n^{-3/4}$. 
The parameter $\kappa$ turns out to be related to the size of the
planted clique through $k \approx n\kappa$.  Hence this argument
can only prove that the SOS hierarchy fails to detect hidden cliques
of size   $k \ll n^{1/4}$. 

In order to improve over this,  and establish Theorem \ref{thm:main}
we prove that matrices  $\tJ_{1, \nu}$
and $\tJ_{2, \nu}$ satisfy certain spectral properties with respect to the subspaces $\bbV_0, \bbV_1, \bbV_2$. 
For instance consider the sum $\tJ_{2, 3}+\tJ_{2, 5}$. For  any $v\in \reals^{\nCe}$
  \begin{align}
    (\tJ_{2, 3}v + \tJ_{2, 5}v)_{\{i, j\}} &= \sum_{k<\ell} p^2(g_{ik}g_{i\ell} + g_{jk}g_{j\ell})v_{\{k, \ell\}}\\
    &= u_i + u_j,
  \end{align}
  where we let $u_i \equiv = \sum_{k<\ell}p^2(g_{ik}g_{i\ell})v_{\{k, \ell\}}$. It follows that
  $(\tJ_{2, 3}v + \tJ_{2, 5})v \in \bbV_0\oplus\bbV_1$ hence $\proj_2 (\tJ_{2, 3} + \tJ_{2, 5}) = 0$.
  By taking transposes we obtain that $(\tJ_{2, 2} + \tJ_{2, 4})\proj_2 = 0$. 
  In a similar fashion we obtain that $\proj_2 (\sum_{\nu}\tJ_{1, \nu}) = (\sum_{\nu}\tJ_{1, \nu})\proj_2 = 0$. 
  See Lemmas \ref{lem:normbndeta1}, \ref{lem:normbndeta2} for formal statements and proofs. 

  Using these observations and \myeqref{eq:triangle} we obtain that $\norm{\proj_2 Q\proj_2} \ls \kappa^4 n$,
  while for any other pair $(a, b)\in \{0, 1, 2\}^2$ we have that $\norm{\proj_a Q\proj_b} \ls \kappa^4 n^{3/2}$. 
  As noted before, since $\lambda_0 \approx n^2 \kappa^4$, $\lambda_1 \approx n\kappa^3$ and $\lambda_1 \approx \kappa^2$
  whence the condition in \myeqref{eq:22PSDproj} reduces to:
  \begin{align}
  \begin{pmatrix}
    n^2\kappa^4 &0 & 0\\
    0 & n\kappa^3  &0 \\
    0 &0 & \kappa^2
  \end{pmatrix} - \kappa^4
  \begin{pmatrix}
     n^{3/2} &  n^{3/2} &  n^{3/2}\\
     n^{3/2} &  n^{3/2} &  n^{3/2}\\
     n^{3/2} &  n^{3/2} &  n
  \end{pmatrix} &\mge 0.
  \end{align}
The $2, 2$ entry of this matrix inequality yields that $\kappa^2 - \kappa^4 n \gg 0$ or
$\kappa \ll n^{-1/2}$. Considering the $(1, 1)$ entry yields a similar condition. The
key condition is that corresponding to the minor indexed by rows (and columns) $1, 2$:
\begin{align}
  \begin{pmatrix}
    n\kappa^3  & -n^{3/2} \kappa^4\\
    - n^{3/2}\kappa^4 &\kappa^2
  \end{pmatrix}&\mge 0.
\end{align}
This requires that $n\kappa^5 \gg n^3 \kappa^8$ or, equivalently $\kappa\ll n^{-2/3}$. 
Translating this to clique size $k = n\kappa$, we obtain the condition  $k \ll n^{1/3}$. This calculation thus
demonstrates  the origin of the threshold of $n^{1/3}$ beyond which the Meka-Wigderson
witness fails to be positive semidefinite. The counterexample of \cite{barak2014sum}
shows that our estimates are fairly tight (indeed, up to a logarithmic factor).

\section{Proofs}
\label{sec:Proof}

\subsection{Definitions and notations}

Throughout the proof we denote the identity matrix in $m$ dimensions
by $\id_m$, and the all-ones vector by  $\one_m$. 
We let $\qproj_n = \one_n\one_n^\sT/n$ be the projector onto the all
  ones vector $\one_n$, and $\qproj_n^\perp = \id_n - \qproj_n$ its
  orthogonal complement.

The indicator function of property $A$ is denoted by $\ind(A)$. 
The set of first $m$ integers is denoted by $[m] = \{1,2,\dots,m\}$.

As mentioned above, we write $f(n,r,\dots)\gtrsim g(n,r,\dots)$ if \emph{there exists a constant}
$C$ such that $f(n,r,\dots) \ge C\, g(n,r,\dots)$. 
Similarly we write 
$f(n, r, \dots) \gg g(n, r, \dots)$ if, \emph{for any constant} $C$, we have 
$f(n,r,\dots) \ge C\, g(n,r,\dots)$ for all $n$ large enough.
These conditions are always understood to hold uniformly with respect
to the  extra arguments $r,\dots$, provided these  belong to a range
depending on $n$, that will be 
clear from the context. 

We finally use the
shorthand  $\barn \equiv n\log n$.

\subsection{Main technical result and proof of Theorem \ref{thm:main}}

The key proposition is the following which controls the matrices
$\tN_{a,b}$. A set of conditions for the parameters $\ualpha$ is
stated in terms of two  matrices $\barW , W\in \reals^{3\times
  3}$. Below we will develop approximations to these matrices, under the parameter
values of Theorem \ref{thm:main}. This allows to  check easily the
conditions  of Proposition \ref{prop:psddecomp}.
\begin{proposition} \label{prop:psddecomp}
  Consider the symmetric matrices $\barW , W\in \reals^{3\times 3}$, where
  $\barW$ is diagonal, and given by:
  \begin{align}
  \barW_{00} &= \alpha_2 + 2(n-2)\alpha_3 p +
  \frac{(n-2)(n-3)}{2}\alpha_4 p^4 - \frac{n(n-1)}{2}\alpha_2^2\, ,\label{eq:barW00}\\
  \barW_{11} &= \alpha_2 + (n-4)\alpha_3 p - (n-3)\alpha_4 p^4\, , \label{eq:barW11}\\
  \barW_{22} &= \alpha_2 -2\alpha_3 p + \alpha_4 p^4\, , \label{eq:barW22}
  \end{align}
 and $W$ is defined  by: 
  \begin{align}
W_{00} &= C\alpha_3 \barn^{1/2} + C\alpha_4\barn^{3/2} 
+ \frac{C(\alpha_3\barn)^2}{\alpha_1} + \frac{\left(n^{3/2}\alpha_3 p^2 +
2\sqrt{n}\alpha_2 + C\alpha_3 \barn\right)^2}{n(\alpha_2 p -
  \alpha_1^2)}\, ,\label{eq:W00}\\
  W_{01} &= C\alpha_3\barn^{1/2} + C\alpha_4 \barn^{3/2} +\frac{C}{\alpha_1}(\alpha_3 \barn)(C\alpha_3\barn + \sqrt{n}\alpha_2) \nonumber \\
  &\quad+ \frac{1}{n(\alpha_2 p - \alpha_1^2)} (n^{3/2}\alpha_3 p^2 +
  2\sqrt{n}\alpha_2 + C\alpha_3 \barn)(3\alpha_3
  \barn)  \, ,\\
  W_{02} &= C\alpha_3 \barn^{1/2} + C\alpha_4
  \barn^{3/2} + \frac{C(\alpha_3\barn)^2}{\alpha_1} \nonumber\\
  &+ \frac{C}{n(\alpha_2 p - \alpha_1^2)} \left(n^{3/2}\alpha_3 p^2 + 2\sqrt{n}\alpha_2
  + C\alpha_3 \barn\right) \left( \alpha_3 \barn \right)\, , \\
  W_{11} &= C\alpha_3\barn^{1/2} + C\alpha_4 \barn^{3/2} + \frac{2}{\alpha_1} \left( C\alpha_3 \barn  + \sqrt{n}\alpha_2 \right)^2 +
  \frac{C(\alpha_3\barn)^2}{n(\alpha_2 p - \alpha_1^2)}\, , \\
  W_{12} &= C\alpha_3\barn^{1/2} +C \alpha_4 \barn^{3/2}  + \frac{C}{\alpha_1}(\alpha_3\barn)( C\alpha_3
  \barn  + \sqrt{n}\alpha_2) + \frac{C(\alpha_3\barn)^2}{n(\alpha_2 p -
           \alpha_1^2)}\, , \\
 W_{22} &= C\alpha_3 \barn^{1/2} + C\alpha_4 \barn +
  \frac{C(\alpha_3\barn)^2}{\alpha_1} + \frac{C(\alpha_3
  \barn)^2}{n(\alpha_2 p - \alpha_1^2)}\, .\label{eq:W22}
 \end{align}
  Assume the following
  conditions hold for a suitable constant $C$:
  \begin{align}
    \alpha_1 &\ge 2 \alpha_2 p + 2\alpha_2 \barn^{1/2},  \\
    \alpha_2 p^2  &\ge \alpha_1^2\,   ,\label{eq:Assumption2BigPropo}\\
    \barW  &\mge W\, . 
  \end{align}
  Then with probability exceeding 
 $1 - n^{-1}$ 
  all of the following  are true:
  \begin{align}
    \tN_{11} &\mge 0 \, ,\label{eq:11PSD}\\
    \tN_{11}^{-1} &\mle \frac{1}{n(\alpha_2 p - \alpha_1^2)}\qproj_n +
                    \frac{2}{\alpha_1}\qproj_n^\perp \, ,\label{eq:11Bound}\\
    \tN_{22} &\mge \frac{2}{\alpha_1}\tN_{12}^\sT \qproj_n^\perp \tN_{12} +
    \frac{1}{n(\alpha_2 p  - \alpha_1^2)} \tN_{12}^\sT \qproj_n \tN_{12}.  \label{eq:22PSD}
  \end{align}
\end{proposition}

The next two lemmas develop simplified expressions for  matrices $\barW$, $W$ under the parameter
choices of Theorem \ref{thm:main}. 
\begin{lemma} 
    \label{lem:barWsimple}
    Setting $(\ualpha, p)$ as in Theorem \ref{thm:main}, there exists
    $\delta_n = \delta_n(\kappa, p)$ with $\delta_n(\kappa, p)\to 0$
    as $n\to \infty$, such that
    \begin{align}
      \abs{\barW_{00}  - \frac{2n^2\kappa^4}{p^2}} &\le \delta_n
                                                     \barW_{00}\, , \\
      \abs{\barW_{11} - \frac{n\kappa^3}{p^2}} &\le \delta_n
                                                 \barW_{11}\, , \\
      \abs{\barW_{22} - \frac{2\kappa^2}{p}}  &\le \delta_n
                                                \barW_{22}\, .
    \end{align}
\end{lemma}
  \begin{lemma} 
    \label{lem:Wsimple}
    Setting $(\ualpha, p)$ as in Theorem \ref{thm:main}, there exists
    $\delta_n = \delta_n(\kappa, p)$ with $\delta_n(\kappa,p) \to
    0$ as $n\to \infty$, such that,    for some absolute constant $C$,
    \begin{align}
      \abs{W_{00} - \frac{n^2\kappa^4}{p^2} }  &\le \delta_n W_{00}\,
      ,\\
      \abs{W_{11} - C\frac{\kappa^4 \barn^{3/2}}{p^6}  } &\le
      \delta_n W_{11} \, ,\\
      \abs{W_{22} - C\frac{\kappa^3\sqrt{\barn}}{p^3} - C\frac{\kappa^5 \barn^2}{p^6} } &\le
      \delta_n W_{22}\, ,\label{eq:W22Approx}
  \end{align}
and, for every $a\ne b\in\{0, 1, 2\}$,
    \begin{align}
     \abs{W_{ab} - C
      \frac{\kappa^4\barn^{3/2}}{p^6} } &\le \delta_n W_{ab},
    \end{align}

  \end{lemma}
With Proposition \ref{prop:psddecomp} and the auxiliary Lemmas \ref{lem:Wsimple},
\ref{lem:barWsimple} in hand, the proof of Theorem
\ref{thm:main} is straightforward. 
\begin{proof}[Proof of Theorem \ref{thm:main}]
As noted in Section \ref{sec:Strategy} it suffices to prove that
$\tN\succeq 0$.
  By taking the Schur complement with respect to $\tN_{11}$, we
  obtain
that $\tN\succeq 0$  if and only if
  \begin{align}
\tN_{11} \mge 0\;\;\;\;\;
\mbox{and}\;\;\;\;\;\;\; 
\tN_{22} \mge
      \tN_{12}^\sT\tN_{11}^{-1}\tN_{12}\, .\label{eq:ConditionstN}
\end{align}
  Suppose that the  conditions of Proposition \ref{prop:psddecomp} are
  verified under the  values of $\ualpha, p$ specified as in Theorem
  \ref{thm:main}. Then we have $\tN_{11}\succeq 0$ by
  Eq.~(\ref{eq:11PSD}). Further by Eqs.~(\ref{eq:11Bound}) and
  (\ref{eq:22PSD}), we have
  \begin{align}
    \tN_{22} &\mge \tN_{12}^\sT\left( \frac{2}{\alpha_1} \qproj_n^\perp +
    \frac{1}{n(\alpha_2 p  - \alpha_1^2)} \qproj_n \right)  \tN_{12} \\
    &\mge \tN_{12}^\sT\tN_{11}^{-1}\tN_{12}\, ,
  \end{align}
which yields the desired (\ref{eq:ConditionstN}).

  We are now left to verify  the conditions of Proposition \ref{prop:psddecomp}. 
  To begin, we verify that
  $\alpha_1 \gs 2\alpha_2 p + 2\alpha_2 \barn^{1/2}$.
  This condition is satisfied if:
  \begin{align}
    p &\gg \kappa\barn^{1/2}\, . 
  \end{align}
  For this, it suffices that
  \begin{align}
    (\kappa\log n)^{1/4}n^{1/6} &\gg \kappa \barn^{1/2}\, .\\
    \text{ or } \kappa&\ls n^{-4/9}(\log  n)^{-1/3}.
  \end{align}
  Since $\kappa \le n^{-2/3}$, this is true. 

  The condition $\alpha_2 p - \alpha_1^2 \ge 0$ holds since
  $\alpha_2 p - \alpha_1^2 = 2\kappa^2  - \kappa^2 = \kappa^2 >0$.

  It remains to check that $\barW \mge W$. By Sylvester's criterion, 
  we need to verify that:
  \begin{align}
    \barW_{00} - W_{00}& >0\, , \label{eq:sylv1}\\
    \begin{vmatrix}
      \barW_{00} - W_{00} & -W_{01}\\
      -W_{01} &\barW_{11} - W_{11}
    \end{vmatrix} & > 0 \, ,\label{eq:sylv2}\\
    \begin{vmatrix}
      \barW_{00} - W_{00} & -W_{01} & -W_{02}\\
      -W_{01} &\barW_{11} - W_{11} & -W_{12} \\
      -W_{02} &-W_{12} &\barW_{22} - W_{22}
    \end{vmatrix} & > 0\, . \label{eq:sylv3}
  \end{align}
It suffices to check the above values using the simplifications provided
by Lemmas \ref{lem:barWsimple} and \ref{lem:Wsimple} respectively as
follows. Throughout, we will assume that $n$ is large enough, and write
$\delta_n$ for a generic sequence such that $\delta_n\to 0$ uniformly
over $\kappa \in [\log n/n, c^{-4} n^{-2/3}/\log n]$, $p\in [c (\kappa\log n)^{1/4} n^{1/6},1]$.

For \myeqref{eq:sylv1}, using Lemmas \ref{lem:barWsimple}
and \ref{lem:Wsimple} we have that:
\begin{align}
  \barW_{00} - W_{00} &\ge \frac{n^2 \kappa^4}{2p^2} ,
\end{align} 
Hence , $\barW_{00} - W_{00} \ge {n^2 \kappa^4}/{2p^2} > 0$
for large enough $n$. 

For \myeqref{eq:sylv2} to hold we need:
\begin{align}
  (\barW_{00} - W_{00})(\barW_{11} - W_{11}) - W_{01}^2 > 0.
\end{align}
By Lemmas \ref{lem:barWsimple} and \ref{lem:Wsimple} we have:
\begin{align}
  \barW_{11} - W_{11} &\ge \frac{n\kappa^3}{p^2}(1-\delta_n) -
  \frac{C\kappa^4\barn^{3/2}}{p^6}(1+\delta_n).
\end{align}
The ratio of the two terms above is (up to a constant) given by
$ p^4 /(\kappa n^{1/2}(\log n)^{3/2})\to \infty$, hence for $n$ large enough we
have $\barW_{11} - W_{11} \ge n\kappa^3/2p^2$. Thus 
Eq.~(\ref{eq:sylv2}) holds if
\begin{align}
  \left( \frac{n^2 \kappa^4}{p^2} \right)\left( \frac{n\kappa^3}{p^2} \right)
  &\gg \left(  \frac{\kappa^4 \barn^{3/2}}{p^{6}} \right)^2\\
  \text{ or } p^8 &\gg \kappa (\log n)^3.
\end{align}
However as we set $ p \gs  (\kappa\log n)^{1/4}n^{1/6}$, this is satisfied
for $n$ large. Indeed this implies that:
\begin{align}
    \begin{vmatrix}
      \barW_{00} - W_{00} & -W_{01}\\
      -W_{01} &\barW_{11} - W_{11}
    \end{vmatrix} & \ge \frac{n^3 \kappa^7}{2p^4}\, .
\end{align}

Consider now \myeqref{eq:sylv3}. Expanding the determinant along the third
column
\begin{align}
  (\barW_{22} - W_{22}) \begin{vmatrix}
      \barW_{00} - W_{00} & -W_{01}\\
      -W_{01} &\barW_{11} - W_{11}\\
  \end{vmatrix} 
  +W_{12} \begin{vmatrix}
      \barW_{00} - W_{00} & -W_{01}\\
      -W_{02} &-W_{12}
  \end{vmatrix}
  - W_{02}\begin{vmatrix}
      -W_{01} &\barW_{11} - W_{11}\\
      -W_{02} &-W_{12}
  \end{vmatrix} > 0 \, .
\end{align}
We start by noting that, for all $n$ large enough, 
\begin{align}
  \barW_{22} - W_{22} &\ge \frac{3\kappa^2}{2p}. 
\end{align}
Indeed, by Lemma \ref{lem:barWsimple} and \ref{lem:Wsimple}, to prove
this claim it is sufficient to show that
\begin{align}
  \frac{\kappa^2}{p} &\ge  C \left(\frac{\kappa^5 \barn^2}{p^6}  +
    \frac{\kappa^3 \barn^{1/2}}{p^3}\right), 
\end{align}
for a large enough constant $C$
or:
\begin{align}
 p &\ge C \max\big(  n^{2/5}\kappa^{3/5} (\log n)^{2/5},
 \kappa^{1/2}(n\log n)^{1/4}\big)\,
\end{align}
This is satisfied when we choose $p \ge c(\kappa\log n)^{1/4} n^{1/6}$ when
we choose $c$ a large enough constant.
Along with the argument for the second condition
above, this implies that:
\begin{align}
  (\barW_{22} - W_{22}) \begin{vmatrix}
      \barW_{00} - W_{00} & -W_{01}\\
      -W_{01} &\barW_{11} - W_{11}
    \end{vmatrix} &\ge \frac{n^3 \kappa^9}{2p^5},\label{eq:FirstDet}
\end{align}
for large enough $n$. 

We now consider the second term. 
Let $w \equiv C\kappa^4\barn^{3/2}/p^{6}$. Then by Lemmas \ref{lem:barWsimple} and \ref{lem:Wsimple}, 
 for all $n$ large enough:
\begin{align}
0\le   -W_{12} \begin{vmatrix}
      \barW_{00} - W_{00} & -W_{01}\\
      -W_{02} &-W_{12}
    \end{vmatrix} &\le \frac{3}{2}w^2\left( \frac{n^2 \kappa^4}{p^2} +
    w\right)\\
    &\le \frac{2n^2 \kappa^4 w^2}{p^2},\label{eq:SecondDet}
\end{align}
as $n^2 \kappa^4/ p^2 > 2w$ whenever $p \ge (\log n)^{3/8}n^{-1/8}$. As we have
$p \ge n^{-1/12}$ this is satisfied. 

Similarly, for the third term 
\begin{align}
   0\le W_{02}\begin{vmatrix}
      -W_{01} &\barW_{11} - W_{11}\\
      -W_{02} &-W_{12}
    \end{vmatrix} &\le \frac{3w^2}{2}\left( w + \frac{n\kappa^3}{p^2}
    \right). 
\end{align}
The second term in the parentheses above dominates when $p \ge \kappa^{1/4}(\log
n)^{3/8}n^{1/8}$ which holds as we keep $p \ge c (\kappa\log n)^{1/4}n^{1/6}$. 
Hence:
\begin{align}
   W_{02}\begin{vmatrix}
      -W_{01} &\barW_{11} - W_{11}\\
      -W_{02} &-W_{12}
    \end{vmatrix} &\le \frac{2n\kappa^3 w^2}{p^2}.\label{eq:ThirdDet}
\end{align}
Thus, using Eqs.~(\ref{eq:FirstDet}),  (\ref{eq:SecondDet}),
(\ref{eq:ThirdDet}),
we conclude that \myeqref{eq:sylv3} holds if
\begin{align}
  \frac{n^3 \kappa^9}{2 p^5} &\ge \frac{2n^2 \kappa^4 w^2}{p^2}  +
  \frac{2n\kappa^3 w^2}{p^2} \\
  &= \frac{2(1+n\kappa)n\kappa^3 w^2}{p^2}.
\end{align}
For this, it suffices that:
\begin{align}
  \frac{n^3 \kappa^9}{p^5} &\gs \frac{n^2\kappa^4 w^2}{p^2}, 
\end{align}
or, equivalently,  $p^{9} \ge c_1 n^2\kappa^3(\log n)^3$ for an 
appropriate $c_1$ large enough. 
This holds under the stated condition  $p\ge c(\kappa\log n)^{1/4} n^{1/6}$ 
provided $c$ is large enough.
This completes the proof
of Theorem \ref{thm:main}.
\end{proof}
The proofs of Lemma \ref{lem:Wsimple} and \ref{lem:barWsimple}
follow by a simple calculation and are given in Section
\ref{subsec:Wlemmas}. 

 Our key technical result is Proposition
\ref{prop:psddecomp}. 
Its proof is organized as follows. We analyze
the expectation matrices $\E\{\tN_{22}\}$, $\E\{\tN_{12}\}$ in
Section \ref{subsec:expectation}. 
We then control the random components 
$\tN_{11}-\E\{\tN_{11}\}$ in Section \ref{subsec:deviationM11},
$\tN_{12}-\E\{\tN_{12}\}$ in Section \ref{subsec:deviationM12}, and
$\tN_{22}-\E\{\tN_{22}\}$ in Section \ref{subsec:deviationM22}.
The application of the moment method to these deviations
requires the definition of various specific graph primitives, 
which we isolate in Section \ref{subsec:prelims} for easy reference. 
Finally, we combine the results to establish Proposition 
\ref{prop:psddecomp} in Section \ref{subsec:proofpsddecomp}.

\subsection{Proofs of Lemmas \ref{lem:Wsimple} and
  \ref{lem:barWsimple}}

\label{subsec:Wlemmas}
\begin{proof}[Proof of Lemma \ref{lem:Wsimple}]
  Recall that $W_{00}$ is defined as:
  \begin{align}
W_{00} &= \alpha_3 \barn^{1/2} + C\alpha_4\barn^{3/2} 
+ \frac{C(\alpha_3\barn)^2}{\alpha_1} +\frac{\left(n\sqrt{n}\alpha_3 p^2 +
2\sqrt{n}\alpha_2 + 3\alpha_3 \barn\right)^2}{n(\alpha_2 p -
  \alpha_1^2)}.
  \end{align}

  Firstly, since $p \ge c(\kappa\log n)^{1/4}n^{1/6}$, and $n\kappa \ge \log n$, 
  we have that $p\ge n^{-1/12}$ asymptotically. Hence:
  \begin{align}
    \frac{n\sqrt{n}\alpha_3 p^2}{\alpha_3 \barn} &=
    \frac{p^2\sqrt{n}}{\log n}\to \infty.
  \end{align}
  Similarly:
  \begin{align}
    \frac{n\sqrt{n}\alpha_3 p^2 }{\sqrt{n} \alpha_2} &= \frac{n\kappa}{2} \to
    \infty.
  \end{align}
  Also:
  \begin{align}
    \frac{\alpha_4 \barn^{3/2}}{n^3 \alpha_3^2 p^4/n(\alpha_2
    p - \alpha_1^2)} &\ls \frac{\kappa^4 \barn^2}{(n^3\kappa^6/n\kappa^2p^2)} \\
      &= \frac{\log^2 n}{ n p^4} \le \frac{\log^2 n}{n^{7/8}} \to 0.\\
      \frac{(\alpha_3\barn)^2/\alpha_1)}{n^3\alpha_3^2
      p^2/n(\alpha_2 p - \alpha_1^2)} &\ls \frac{\kappa\log^2 n}{p^4}\le
      \frac{\kappa\log ^2 n}{\sqrt{n}} \to 0. \\
      \frac{\alpha_3 \barn^{1/2}}{ \alpha_4 \barn^{3/2}} &\le \frac{\kappa^3 /p^3}{ \kappa^4 \barn/p^3} =
	\frac{p^3}{\kappa \barn} \to 0.  
      \end{align}
  Hence the term $(n\sqrt{n}\alpha_3 p^2)^2/n(\alpha_2 p - \alpha_1)$ is
  dominant in $W_{00}$ and the first claim of the lemma follows. 

  For $W_{01}$ we have the equation:
  \begin{align}
  W_{01} &= \alpha_3\barn^{1/2} + C \alpha_4\barn^{3/2} +
  \frac{C}{\alpha_1}(\alpha_3 \barn)(\alpha_3(\barn + \sqrt{n}\alpha_2)) \nonumber \\
  &\quad+ \frac{1}{n(\alpha_2 p - \alpha_1^2)} (n\sqrt{n}\alpha_3 p^2 +
  2\sqrt{n}\alpha_2 + \alpha_3 \barn)(\alpha_3
  \barn).
  \end{align}
  It suffices to check that $C\alpha_4 \barn^{3/2}$ 
  is the dominant term. By the
  argument in $W_{00}$ we already have that the first term is negligible. 
  Further since, $\alpha_3 \barn/\sqrt{n}\alpha_2 =  \kappa\sqrt{n}\log n/
  p^2 = (\kappa\log n)^{1/2} n^{1/6} \to 0 $, to prove that the third
  term is negligible, it suffices that
  \begin{align}
    \frac{(\alpha_3 \barn)(\sqrt{n}\alpha_2)}{\alpha_1 \alpha_4 \barn^{3/2}} &\le \frac{\kappa^5 p^{-3}}{\kappa^5 p^{-6} \sqrt{\log n}} =
      \frac{p^2}{\sqrt{\log n}} \to 0.
  \end{align}
  By the estimates in $W_{00}$ the fourth term is negligible if:
  \begin{align}
    \frac{(n\sqrt{n}\alpha_3 p^2 ) (\alpha_3 \barn)}{n(\alpha_2 p - \alpha_1^2)
      \alpha_4 \barn^{3/2}} &\to 0 \\
      \text{i.e. }  \frac{n^{5/2}\log n p^-4 \kappa^6}{n^{5/2}\log
      n^{3/2} p^{-6}\kappa^6} &= \frac{p^2}{\sqrt{\log n}} \to 0. 
  \end{align}
  This implies the claim for $W_{01}$. The calculation for $W_{02}$ 
  and $W_{12}$ is similar. 

  We now consider $W_{11}$ given by:
  \begin{align}
  W_{11} &= \alpha_3\barn^{1/2} + C \alpha_4 \barn^{3/2} \nonumber \\
  &+ \frac{C}{\alpha_1} \left( C\alpha_3 \barn  + \sqrt{n}\alpha_3 p^2
  + 2\alpha_2 \right)^2 + \frac{C(\alpha_3\barn)^2}{n(\alpha_2 p
  - \alpha_1^2)}.
  \end{align}
  As in $W_{00}$, the first term is negligible. For the third term, first we
  note that $\alpha_3 \barn /\alpha_2 = (\kappa\log n) n/ p^2 \ge \log ^2
  n\to \infty$. Hence to prove that the third term is negligible, it suffices
  that:
  \begin{align}
    \frac{(\alpha_3 \barn)^2}{\alpha_1 \alpha_4 \barn^{3/2}} &\le
    {\kappa \sqrt{\barn}} \to 0.
  \end{align}
  The final term in $W_{11}$ is negligible by the same argument, since
  $n(\alpha_2 p - \alpha_1^2) = n\kappa^2 \ge \alpha_1$.

  $W_{22}$ is given by:
  \begin{align}
 W_{22} &= \alpha_3 \barn^{1/2} + C\alpha_4 \barn \nonumber \\
 &\quad+
  \frac{C(\alpha_3\barn)^2}{\alpha_1} + \frac{C(\alpha_3
  \barn)^2}{n(\alpha_2 p - \alpha_1^2)}.
  \end{align}
Since $n(\alpha_2 p - \alpha_1^2) = n\kappa^2 \ge \alpha_1$ it is easy to see
that the third term dominates the fourth above. To see that the first dominates
the second, it suffices that their ratio diverge i.e.
\begin{align}
  \frac{\alpha_3 \barn^{1/2}}{\alpha_4 \barn} &= \frac{p^3}{\kappa
    \sqrt{\barn}} \\
    &\ge \frac{p^3}{\kappa \log n \sqrt{n}} \\ 
    &= c^3 (\kappa\log n)^{1/8} n^{1/4} \to\infty,
\end{align}
as $\kappa \ge 1/n$. Thus we have that the first and third
terms dominate the contribution for $W_{22}$. 
This completes the proof of the lemma. 
\end{proof}

\begin{proof}[Proof of Lemma \ref{lem:barWsimple}]
  $\barW_{00}$ is given by:
  \begin{align}
    \barW_{00} &= \alpha_2 + 2(n-2)\alpha_3 p +
  \frac{(n-2)(n-3)}{2}\alpha_4 p^4 - \frac{n(n-1)}{2}\alpha_2^2.
  \end{align}
  It is straightforward to check that the third and fourth terms
  dominates the sum above i.e.:
  \begin{align}
    \frac{\barW_{00}}{\frac{(n-2)(n-3)}{2}\alpha_4 p^4 - \frac{n(n-1)}{2}\alpha_2^2} \to 1.
  \end{align}
  Further we have:
  \begin{align}
    \frac{(n-2)(n-3)}{2}\alpha_4 p^4 - \frac{n(n-1)}{2}\alpha_2^2 &= (1+ \delta_n) \frac{2n^2 \kappa^4}{p^2},
  \end{align}
  for some $\delta_n\to 0$. The claim for $\barW_{00}$ then follows. 

  The claims for $\barW_{11}$ and $\barW_{22}$ follow in the same fashion as above
  where we instead use the following, adjusting $\delta_n$ appropriately:
  \begin{align}
    \frac{\barW_{11}}{n\alpha_3 p} &= \frac{\alpha_2 + (n-4)\alpha_3 p - (n-3)\alpha_4 p^4 }{n \alpha_3 p}
    \to 1 \\
    \frac{\barW_{22}}{\alpha_2} &= \frac{ \alpha_2 -2\alpha_3 p + \alpha_4 p^4 }{\alpha_2}\to 1.
  \end{align}
\end{proof}

\subsection{Graph definitions and moment method} \label{subsec:prelims}

In this section we define some family of  graphs that will be useful in the moment
calculations of Sections \ref{subsec:deviationM11},  \ref{subsec:deviationM22} 
and \ref{subsec:deviationM12}. We then state and prove a moment method
lemma,
that will be our basic tool for controlling the norm of random matrices.

\begin{definition}
  A \emph{cycle} of length $m$ is a graph $D=(V, E)$ with vertices
  $V = \{v_1, \dots v_m\}$ and edges $E= \{\{v_{i}, v_{i+1}\}:\;  i\in [m]\}$
where addition is taken modulo $m$.
\end{definition}

\begin{definition}
  A \emph{couple} is an ordered pair of vertices $(u, v)$ where we
  refer to the first vertex in the couple as the \emph{head} and the second as the
  \emph{tail}. 
\end{definition}

\begin{definition}
  A \emph{bridge} of length $2m$ is a graph $B=(V, E)$ with vertex set
  $V = \{u_{i}, v_{i}, w_{i} :\; i\in [m]\}$, and edges
  $E=\{\{u_i, v_i \}, \{u_i, w_i\}, \{u_{i+1}, v_i\}, \{u_{i+1},
  w_i\}:\, i\in [m]\}$
  where addition above is modulo $m$. We regard $(v_i, w_i)$ for $i\in
  [m]$
  as couples in the bridge. 
\end{definition}

\begin{definition}
 A \emph{ribbon} of length $m$ is a graph $R=(V, E)$
  with vertex set $V =\{u_1 \dots u_{m}, \allowbreak v_1 \dots v_m\}$ and 
  edge set $E = \{\{u_{i}, u_{i+1}\}, \{u_{i}, v_{i+1}\}, 
  \{v_{i}, u_{i+1}\}, \{v_{i}, v_{i+1}\}: i\in [m]\}$  where addition
  is modulo $m$. Further we call the subgraph induced by the 4-tuple $(u_i, v_i, u_{i+1}, v_{i+1})$
  a \emph{face} of the ribbon and we call the ordered pairs $(u_i, v_i)$,
  $i\in [m]$ couples of ribbon. 
\end{definition}

Each face of the ribbon has $4$ edges, hence there are $\binom{4}{\eta}$ ways
to remove $4-\eta$ edges from the face. We define a ribbons of class $\eta$, type $\nu$
and length $2m$ as follows.

\begin{definition}
For $1\le \eta\le 4 $ and $1\le \nu \le \binom{4}{\eta}$, we define a
ribbon of length $2m$, class $\eta$ and type $\nu$ to be the graph
obtained from a ribbon of length $2m$ by keeping $\eta$ edges in each face of 
the ribbon, so that the following happens. The subgraphs induced by the tuples $(u_{2i-1}, v_{2i-1}, u_{2i}, v_{2i})$ and
$(u_{2i+1}, v_{2i+1}, u_{2i}, v_{2i})$ for $i\ge 1$ are faces of class $\eta$
and type $\nu$ as shown in Table \ref{tab:ribbondef}. 

For brevity, we write $(\eta, \nu)$-ribbon to denote a ribbon of class
$\eta$ and type $\nu$. 
\end{definition}

\begin{definition}
  A  $(\eta, \nu)$-star ribbon  $S=(V, E)$ of length $2m$ is a graph formed from a $(\eta, \nu)$-ribbon
  $R(V', E')$ of length
  $2m$ by the following process. For each face $(u_i, v_i, u_{i+1}, v_{i+1})$ we
  identify \emph{either} the vertex pair $(u_i, u_{i+1})$ or the pair $(v_{i},
  v_{i+1})$ and delete the self loop formed, if any, from the edge set. Note here
  that the choice of the pair
  identified can differ across faces of $R$. 

We let $\cS_{\eta, \nu}^{m}$ denote
  this collection of $(\eta, \nu)$-star ribbons.
\end{definition}

\begin{definition}
  A labeled graph is a pair $(F=(V, E), \ell)$ where $F$ is a graph and
  $\ell :V\to [n]$ maps
  the vertices of the graph to labels in $[n]$. We define a \emph{valid labeling}
  to be one that satisfies the following conditions:
  \begin{enumerate}
    \item Every couple of vertices $(u, v)$ in the graph satisfies $\ell(u) <
      \ell(v)$.
    \item For every edge $e = \{v_1, v_2\}\in E$, $\ell(v_1) \ne \ell(v_2)$. 
  \end{enumerate}
  A labeling of $F$ is called \emph{contributing} if, in addition to being
  valid, the following happens. For every edge $e=\{u, v\}\in E$,
  there exists an edge
  $e'= \{u', v'\}\ne e$ such that $\{\ell(u), \ell(v)\} = \{\ell(u'),
  \ell(v')\}$. In other words, a
  labeling is contributing if it is valid and has the property that every
  labeled edge occurs \emph{at least twice} in $F$.  
\end{definition}

\begin{remark}
  Suppose $F$ is one of the graphs defined above and $C$ is a face of $F$. 
  We write, with slight abuse of notation, $C\subseteq F$ to denote ``a face $C$
  of the graph $F$''. Furthermore, to lighten notation, we will often write
  $e\in F$  for an edge $e$ in the graph $F$. 
\end{remark}

\begin{definition}
  Let $\frakL(F)$ denote the set of valid labelings of a graph $F=(V, E)$ and
  $\frakL_2(F)$ denote the set of contributing labelings. Further, we define
  \begin{align}
    v_*(F) &= \max_{\ell\in \frakL_2(F)} \range(\ell)
  \end{align}
  where $\range(\ell) = \{i\in [n]: i= \ell(u), u \text{ is a vertex in }F \}$. 
\end{definition}

\begin{table}[h]
  \centering
  \begin{tabular}{lccc}
    \hline Figure & Ribbon class($\eta$) & Ribbon type($\nu$) & Typical norm \\
    \hline
    \parbox[c]{0.5em}{\includegraphics[width=1.3cm]{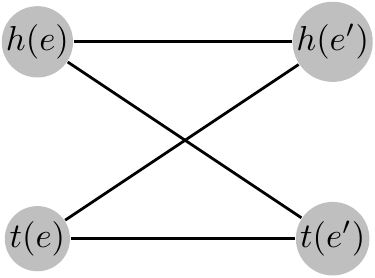}} & 4 & 1 & $\barn$\\
\hline
    \parbox[c]{0.5em}{\includegraphics[width=1.3cm]{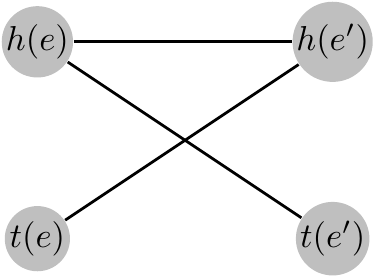}} & 3 & 1& $\barn$\\
    \parbox[c]{0.5em}{\includegraphics[width=1.3cm]{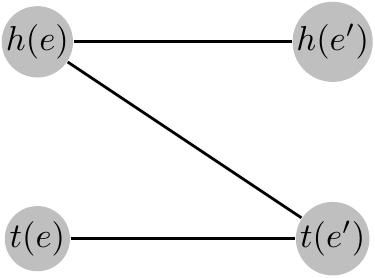}} & 3 & 2& $\barn$\\
    \parbox[c]{0.5em}{\includegraphics[width=1.3cm]{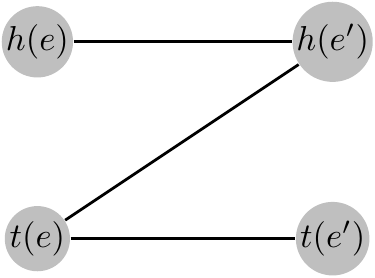}} & 3 & 3& $\barn$\\
    \parbox[c]{0.5em}{\includegraphics[width=1.3cm]{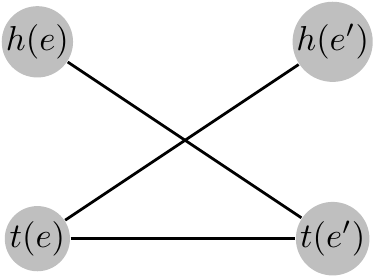}} & 3 & 4& $\barn$\\
\hline
    \parbox[c]{0.5em}{\includegraphics[width=1.3cm]{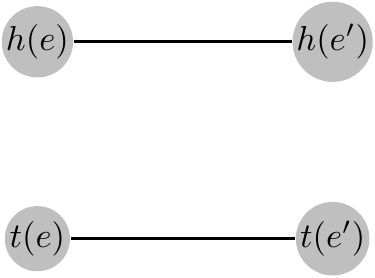}} & 2 & 1& $\barn$\\
    \parbox[c]{0.5em}{\includegraphics[width=1.3cm]{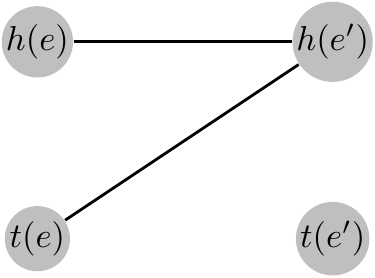}} & 2 & 2& $\barn^{3/2}$\\
    \parbox[c]{0.5em}{\includegraphics[width=1.3cm]{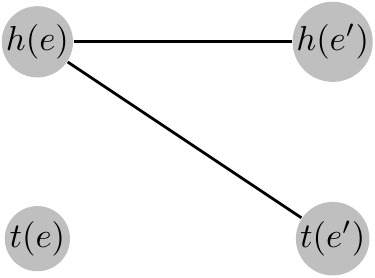}} & 2 & 3&$\barn^{3/2}$\\
    \parbox[c]{0.5em}{\includegraphics[width=1.3cm]{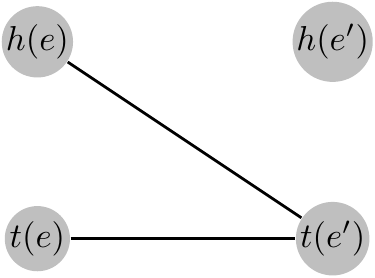}} & 2 & 4& $\barn^{3/2}$\\
    \parbox[c]{0.5em}{\includegraphics[width=1.3cm]{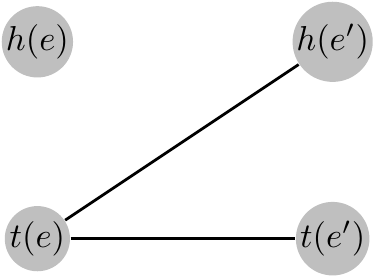}} & 2 & 5& $\barn^{3/2}$\\
    \parbox[c]{0.5em}{\includegraphics[width=1.3cm]{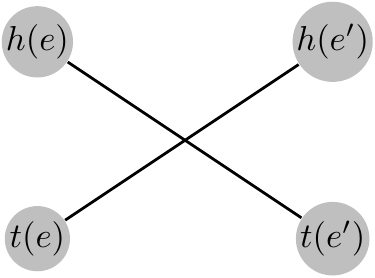}} & 2 & 6& $\barn$\\
\hline
    \parbox[c]{0.5em}{\includegraphics[width=1.3cm]{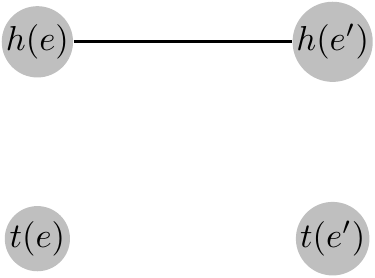}} & 1 & 1& $\barn^{3/2}$\\
    \parbox[c]{0.5em}{\includegraphics[width=1.3cm]{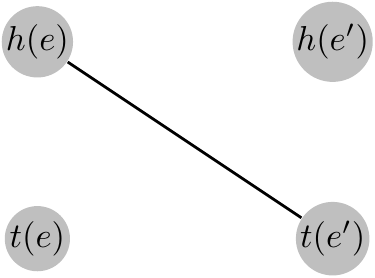}} & 1 & 2& $\barn^{3/2}$\\
    \parbox[c]{0.5em}{\includegraphics[width=1.3cm]{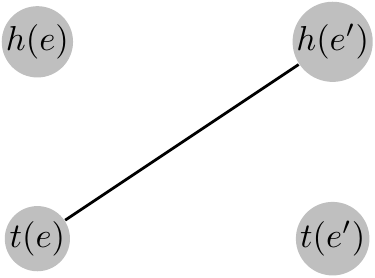}} & 1 & 3& $\barn^{3/2}$\\
    \parbox[c]{0.5em}{\includegraphics[width=1.3cm]{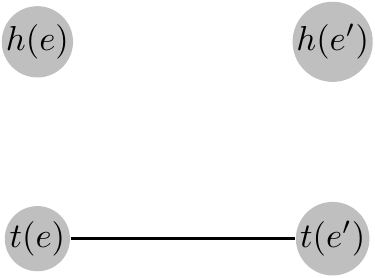}} & 1 & 4& $\barn^{3/2}$ \\
    \hline
  \end{tabular}
  \caption{Definition of the different ribbon classes and types.}
  \label{tab:ribbondef}
\end{table}

The following is a simple and general moment method lemma.
\begin{lemma}\label{lem:trpow}
  Given a matrix $X \in \reals^{m'\times n'}$, suppose
  that there exist constants $c_1$, $c_2, c_3, c_4 , c_5\ge 0$ satisfying $c_2
  \ge c_4$ and for any integer $r>0$:
  \begin{align}
    \E\Tr\{(X^{\sT}X)^{r}\} &\le \binom{n}{c_1 r + c_2} (c_5)^{2r} (c_1r +
    c_2)^{c_3 r + c_4}.
  \end{align}
  Then, for every $n$ large enough, with probability exceeding $1 -
  n^{-(\Gamma -c_2 )/2}$ we have that 
  \begin{align}
  \norm{X}_2 &\le c_4\sqrt{\exp(c_1\Gamma)n^{c_1}(\log n)^{c_3 - c_1}}.  
  \end{align}
\end{lemma}
\begin{proof}
  By rescaling $X$ we can assume that $c_5=1$. 
  Since $\Tr\{(X^\sT X)^{2r}\} = \sum_{i}(\sigma_i(X))^{2r}$ where $\sigma_i(X)$
  are the singular values of $X$ ordered $\sigma_1 (X)\ge \sigma_2 (X) \dots
  \sigma_N(X)$, we have that:
  \begin{align}
    \norm{X}^{2r}_2 &=\sigma_1(X)^{2r} \le \Tr\{(X^\sT X)^{r}\}.
  \end{align}
  Then, by Markov inequality  and the given assumption:
  \begin{align}
    \P\left\{ \norm{X}_2 \ge t \right\}&\le \P\left\{ \Tr\{(X^\sT X)^{2r}\} \ge 
    t^{2r} \right\} \\
    &\le t^{-2r}\E\Tr\left\{(X^\sT X)^{2r}\right\} \\
    &\le \binom{n}{c_1r +c_2} (c_1 r + c_2)^{c_3 r+ c_4}.
  \end{align}
  Using $\binom{n}{k} \le (ne/k)^k$ we have:
  \begin{align}
    \P\left\{ \norm{X}_2 \ge t \right\} &\le t^{-2r} (ne)^{c_1 r + c_2} (c_1r +
    c_2)^{(c_3-c_1)r +c_4 - c_2}\\
    &= \exp\left\{ (c_1 r + c_2) (\log n + 1) + ((c_3-c_1)r +c_4 -c_2) \log (c_1 r + c_2) - 2r \log t  \right\}.
  \end{align}
  Setting $r = \lceil(\log n - c_2)/c_1\rceil$ and using $c_2 \ge c_4$ we obtain the bound:
    \begin{align}
      \P\left\{ \norm{X}_2 \ge t \right\} &\le \exp\left\{ \log n (\log n +1 ) +  (c_3/c_1 -1)(\log n)\log\log n  - (\log n -c_2) \log (t^{2/c_1}) \right\} \\
      &\le \exp\left\{ \log n \log \left( ne (\log n)^{c_3/c_1 - 1} \right) -
      (\log n - c_2) \log (t^{2/c_1}) \right\}.
    \end{align}
    We can now set $t = \left\{\exp(\Gamma) n (\log n)^{c_3/c_1 -
    1}\right\}^{c_1/2}$ whereupon the bound on the right hand side is at most $n^{-(\Gamma - c_2)/2}$ 
    for every $n$ large enough. This yields the claim of the lemma.
\end{proof}

The next lemma specialized the previous one to the type of random matrices we
will be interested in.
\begin{lemma} \label{lem:trexp}
  For a matrix $X\in \reals^{m'\times n'}$, suppose there exists a sequence
  of graphs $G_X(r)$ with vertex, edge sets $V_r, E_r$ respectively, a set $\frakL(G_X(r))$ of labelings
  $\ell:V_r\to[n]$ and a constant $\beta>0$ such that:
  \begin{align}\label{eq:basictraceexp}
    \Tr\left\{(X^\sT X)^r\right\} &= \beta^{2r}\sum_{\ell\in \frakL(G_X(r))}
    \prod_{e \in G_X(r)} g_{\ell(e)}, 
  \end{align}
  where, for $e=\{u,v\}$, $\ell(e) = \{\ell(u), \ell(v)\}$. 
  Let $\frakL_2(G_X(r))\subseteq \frakL(G_X(r))$ denote the subset of
  contributing labelings (i.e. the set of labelings
  $\ell\in\frakL(G_X(r))$ such that every labeled edge in $G_X(r)$ is repeated \emph{at least
  twice}). Further define $v(r)$ and
  $v^*(r)$ by:
  \begin{align}
    v(r) &\equiv \abs{V_r}\, ,\\ 
    v_*(r) &\equiv v_*(G_X(r)).
  \end{align}
  Then
  \begin{align}
    \E\Tr\left\{ (X^\sT X)^{r} \right\} &\le \beta^{2r}\abs{\frakL_2(G_X(r))} \\
    &\le \binom{n}{v_*(r)} \beta^{2r} v_*(r)^{v(r)}.
  \end{align}
\end{lemma}
\begin{proof}
  By rescaling $X$ it suffices to show the case $\beta = 1$.
  Taking expectations on either side of \myeqref{eq:basictraceexp} we have
  that:
  \begin{align}
    \E\Tr\left\{ (X^\sT X)^{r} \right\} &= \sum_{\ell\in \frakL(G_X(r))}
    \E\left\{\prod_{e \in G_X(r)} g_{\ell(e)} \right\}. 
  \end{align}
  The variables $g_{\ell(e)}$ are centered and  independent and bounded 
  by 1. Hence the only terms that do not vanish in the summation above
  correspond to labelings $\ell$ wherein every labeled edge occurs at least twice, 
  i.e. precisely when  $\ell\in \frakL_2(G_X(r))$. By the boundedness of
  $g_{\ell(e)}$, the contribution of each non-vanishing term is at most 1, hence
  \begin{align}
    \E\Tr\left\{ (X^\sT X)^{r} \right\} &\le \abs{\frakL_2(G_X(r))}.
  \end{align}
  It now remains to prove that $\abs{\frakL_2(G_X(r)} \le \binom{n}{v_*(r)} v(r)^{v(r)}$.
  By definition, $\ell$ can map the vertices in $V_r$ to at most $v_*(r)$
  distinct labels. There are at most $\binom{n}{v_*(r)}$ distinct ways to pick
  these labels in $[n]$, and at most $v_*(r)^{v(r)}$ ways to assign the
  $v_*(r)$ labels to $v(r)$ vertices, yielding the required bound.
\end{proof}
\begin{lemma}\label{lem:gentrbnd}
  Consider the setting of Lemma \ref{lem:trexp}. If we additionally have
  \begin{align}
    v_*(r) &\le c_1 r + c_2 \\
    v(r) &= c_3 r + c_4,
  \end{align}
  where $c_3 \le 2c_1$ then $\norm{X}_2 \ls \beta\barn^{c_1/2}$ with probability
  at least $1-n^{-5}$.
\end{lemma}
\begin{proof}
  The proof follows by combining Lemmas \ref{lem:trexp} and \ref{lem:trpow}.
\end{proof}

\subsection{The expected values $\E\{\tN_{22}\}$, $\E\{\tN_{12}\}$}
\label{subsec:expectation}

In  this section we characterize the eigenstructure of the expectations
$\E\{\tN_{22}\}$, $\E\{\tN_{12}\}$.
These can be viewed as linear operators on $\reals^{\nCe}$ that are
invariant under the action of permutations\footnote{A permutation
  $\sigma:[n]\to [n]$ acts on $\reals^{\nCe}$ by permuting the indices
in $\nCe$ in the obvious way, namely $\sigma(\{i,j\}) =\{\sigma(i),\sigma(j)\}$.} on $\reals^{\nCe}$.
By Schur's  Lemma \cite{serre1977linear}, their eigenspace decomposition corresponds
to the decomposition $\reals^{\nCe}$ in irreducible representations of
the group of permutations.
This is given by $\reals^{\nCe}= \bbV_0 \oplus \bbV_1 \oplus \bbV_2$,
where
\begin{align}
  \bbV_0 &\equiv \{v \in \reals^{\nCe}: \exists u\in\reals \text{s.t.}
            v_{\{i, j\}} = u \mbox{ for all } i<j\} \\
  \bbV_1 &\equiv \{v \in \reals^{\nCe}: \exists u\in\reals^{n},
           \text{s.t.} \<\one_n, u\> = 0 \text{ and } 
\quad v_{\{i, j\}} = u_i + u_j \mbox{ for all } i<j\}\\
  \bbV_{2} &\equiv (\bbV_{0}\oplus\bbV_{1})^{\perp}.
\end{align}
An alternative approach to  defining the spaces $\bbV_a$ is 
to let $\bbV_0 = \sspan(v_0), \bbV_1 =
\sspan(v_1^i, i=1\dots n), \bbV_2 = \sspan(v_2^{ij}, 1\le i<j\le n)$,
where
\begin{align}
  (v_0)_A &= \sqrt\frac{2}{n(n-1)} \\
  (v_1^i)_A &= \begin{cases}
    \sqrt{\frac{n-2}{n(n-1)}} &\text{if }  A=\{i, \cdot\} \\
    -\frac{2}{\sqrt{n(n-1)(n-2)}} &\text{otherwise.}
  \end{cases} \\
  (v_2^{ij})_A &= \begin{cases}
    \sqrt{\frac{n-3}{n-1}} &\text{if } A= \{i, j\} \\
    -\frac{1}{n-2}\sqrt{\frac{n-3}{n-1}} &\text{if } A = \{i, \cdot\} \text{ or } \{j, \cdot\} \\
    \frac{1}{\binom{n-2}{2}}\sqrt{\frac{n-3}{n-1}} &\text{ otherwise.}
  \end{cases}
\end{align}
Notice that $\dim(\bbV_0) = 1$, $\dim(\bbV_1) = n-1$, $\dim(\bbV_2) = n(n-3)/2$, 
and that $\{v_1^i\}_{i\in [n]}$,  $\{v_1^{i,j}\}_{i,j\in [n]}$ are
overcomplete sets.
For $a\in \{0,1,2\}$, we denote by $V_a$ the matrix whose rows
are given by this overcomplete basis of $\bbV_a$

It is straightforward to check that the two definitions of the
orthogonal  decomposition 
$\reals^{\nCe}=\bbV_0 \oplus \bbV_1 \oplus \bbV_2$ given above coincide.
We let $\proj_a\in \reals^{\nCe\times\nCe}$ denote
the orthogonal projector on the space $\bbV_a$.

The following proposition
gives the eigenstructure of $\E\{\tN_{22}\}$.
\begin{proposition}\label{prop:EM22eig}
 The matrix  $\E\{\tN_{22}\}$ has the following spectral
 decomposition
\begin{align}
\E\{\tN_{22}\} = \lambda_0 \proj_0 + \lambda_1 \proj_1 +
  \lambda\proj_2\, ,
\end{align}
where
\begin{align}
  \lambda_0 &= \alpha_2 + 2(n-2)\alpha_3 p +
  \frac{(n-2)(n-3)}{2}\alpha_4 p^4 - \frac{n(n-1)}{2}\alpha_2^2\, ,
  \\
  \lambda_1 &= \alpha_2 + (n-4)\alpha_3 p - (n-3)\alpha_4 p^4 \, ,\\
  \lambda_2 &= \alpha_2 -2\alpha_3 p + \alpha_4 p^4.
\end{align}
\end{proposition}
\begin{proof}
  It is straightforward to verify that the vectors $v_\ell^A$ defined
  above are eigenvectors of $\E\{\tN_{22}\}$. The eigenvalues are
  then given by $\lambda_\ell = \<v_{\ell}^A, \E\{\tN_{22}\}v_{\ell}^A\>$
  for an arbitrary choice of $A = \{i\}$ or $\{i, j\}$. 
\end{proof}
\begin{remark}
  The above eigenvalues can also be computed using  
  \cite{meka2013association} which relies on the theory of
  association schemes. We preferred to present a direct and
  self-contained derivation.
\end{remark}

We now have a similar proposition for $\E\{\tN_{12}\}\in
\reals^{\binom{[n]}{1}\times \nCe}$. More precisely, we decompose 
$\reals^{\binom{[n]}{1}}$ in $\sspan(\one_m)$ and its orthogonal complement,
and $\reals^{\nCe}=\bbV_0 \oplus \bbV_1 \oplus \bbV_2$ as above.
\begin{proposition}\label{prop:EM12eig}
  The following hold for all $n$ large enough:
  \begin{align}
    \qproj_n^\perp \E\{\tN_{12}\} \proj_0 &= 0\\
    \norm{\qproj_n^\perp \E\{\tN_{12}\}\proj_1}_2 &\le \sqrt{n}\alpha_2\\
    \qproj_n^\perp \E\{\tN_{12}\}\proj_2 &= 0\\
    \norm{ \qproj_n\E\{\tN_{12}\} \proj_0}_2 &\le n^{3/2}\alpha_3 p^2 + 2\sqrt{n}\alpha_2\\
    \qproj_n\E\{\tN_{12}\} \proj_1 &= 0 \\
    \qproj_n\E\{\tN_{12}\}\proj_2 &= 0.
  \end{align}
\end{proposition}
\begin{proof}
  For $A\in \binom{[n]}{1}$ and $B\in \nCe$:
  \begin{align}
    (\E\{\tN_{12}\})_{A, B} &= \begin{cases}
      \alpha_3 p^2 - \alpha_1 \alpha_2 &\text{ if } \abs{A\cap B} = 0 \\
      \alpha_2 - \alpha_1 \alpha_2 &\text{ if } \abs{A\cap B} = 1.
    \end{cases}
  \end{align}
  Recall from the definition of the space $\bbV_1 = \sspan(\{v^A_1\}_{A\in\binom{[n]}{1}})$. 
  We can write $\E\{\tN_{12}\}$ as:
  \begin{align}
    \E\{\tN_{12}\} &= \frac{\binom{n-1}{2}(\alpha_3 p^2 - \alpha_1\alpha_2) + (n-1) (\alpha_2 - \alpha_3p^2) }
    {\sqrt{\binom{n}{2}}}\one_n v_0^\sT 
    + \sqrt{\frac{(n-1)(n-2)}{n}} (\alpha_2 - \alpha_3 p^2) V_1.
  \end{align}
  This implies all but the second and the fourth claims immediately as 
  $V_1\proj_0 = V_1\proj_2 = 0$, $\qproj_n V_1 = 0$ and $\qproj_n^\perp\one_n = 0$. 
  For the second claim, 
  the above decomposition yields:
\begin{align}
  \norm{\qproj_n^\perp \E\{\tN_{12}\}\proj_1   }_2 &= \max _{x \in\bbV_1: \norm{x}_2 \le 1} \norm{\sqrt{\frac{(n-1)(n-2)}{n}}(\alpha_2 - \alpha_3 p^2) V_1 x}_2 \\
  &= \sqrt\frac{(n-1)(n-2)}{n} (\alpha_2 - \alpha_3 p^2)
    \sqrt{\lambda_{\rm max}
(V_1V_1^\sT)}.
\end{align}
Since $\<v^A_1, v^{A'}_1\> = -1/(n-1)$ when $A\ne A'$ and $1$ otherwise,
we have that:
\begin{align}
  V_1V_1^\sT &= \frac{n }{n-1}\, \id_n - \frac{1}{n-1}\, \one_n(\one_n)^\sT,
\end{align}
hence $\lambda_{\rm max}(V_1V_1^\sT) = n/(n-1)$. This implies
that:
\begin{align}
  \norm{\qproj_n^\perp \E\{\tN_{12}\}\proj_1}_2 &=
  \sqrt{n-2} (\alpha_2 - \alpha_3 p^2)\le \sqrt{n} \alpha_2.
\end{align}
For the fourth claim, the expression for $\E\{\tN_{12}\}$ above yields that:
\begin{align}
  \norm{\qproj_n\E\{\tN_{12}\} \proj_0}_2 &= \frac{\binom{n-1}{2}(\alpha_3 p^2 - \alpha_1 \alpha_2) + (n-1) (\alpha_2-\alpha_3p^2) }
  {\sqrt{\binom{n}{2}}} \sqrt{n} \\
  &\le \frac{\binom{n-1}{2}\alpha_3 p^2}{\sqrt{\frac{n-1}{2}}} +
  \frac{(n-1)\alpha_2}{\sqrt{\frac{n-1}{2}}}\\
  &\le n\sqrt{n}\alpha_3 p^2 + 2\sqrt{n}\alpha_2.
\end{align}
\end{proof}

\subsection{Controlling $\tN_{11} - \E\{\tN_{11}\}$} \label{subsec:deviationM11}

The block $\tN_{11}$ is a linear combination of the identity and the
adjacency matrix of $G$. Hence, its spectral properties are well
understood,
since the seminal work of F\"uredi-Koml\'os
\cite{furedi1981eigenvalues}. While the nest proposition could be
proved using these results, we present an self-contained proof for
pedagogical reasons,  as the same argument will be repeated several
times later for more complex examples.
\begin{proposition} \label{prop:deviationM11}
  Suppose that $\ualpha$ satisfies:
  \begin{align}
    \frac{\alpha_1}{2} - \alpha_2 p &\gs
    \alpha_2\barn^{1/2} \, ,\\
\alpha_2p-\alpha_1^2\ge 0\,,\;\;\;\;\;     \alpha_1 &\ge 0. 
  \end{align}
  Then with probability at least $1-n^{-5}$:
  \begin{align}
    \tN_{11} &\mge 0\, ,\\
    \tN_{11}^{-1} &\mle \frac{1}{n(\alpha_1 p - \alpha_1^2)}\qproj_n +
    \frac{2}{\alpha_1} \qproj_n^\perp
  \end{align}
\end{proposition}
\begin{proof}
  First, note that:
  \begin{align}
    \E\{\tN_{11}\} &= (\alpha_1 - \alpha_2 p)\id_n + (\alpha_2 p - \alpha_1^2 )\, n \qproj_n\, .
  \end{align}
  Furthermore, for $A, B\in \binom{[n]}{1}$, $A\neq B$, $(\tN_{11}- \E\{\tN_{11}\})_{A, B} = \alpha_2
  g_{AB}$. Here, we identify elements of $\binom{[n]}{1}$ with elements
  of $[n]$ in the natural way. 
  Thus, expanding $\Tr\left\{ ((\tN_{11} -
  \E\tN_{11})^{\sT}(\tN_{11}  - \E\tN_{11} ))^m \right\}$ we obtain:
  \begin{align}
    \Tr\left\{ \left( (\tN_{11} - \E\{\tN_{11}\})^\sT (\tN_{11} -
    \E\{\tN_{11}\}) \right)^m \right\} &=  \alpha_2^{2m} \sum_{A_1 \dots
      A_m, A'_1 \dots A'_m}\prod_{\ell=1}^{m} g_{A_\ell A'_\ell}
    g_{A_{\ell+1}A'_\ell},
  \end{align}
  where we set $A_{m+1}\equiv A_1$. 
  Let $D(m)$ be a cycle of length $2m$, $V_D, E_D$ be its vertex and edge
  sets respectively, and $\ell$ be a labeling 
  that assigns to the vertices labels $A_1, A'_1, A_2, A'_2 \dots A_m, A'_m$ in
  order. Then the summation over indices $A_1 \dots A_m,
  A'_1 \dots A'_m$ can be expressed as a sum over such labelings of the
  cycle $D(m)$, i.e.:
  \begin{align}
    \Tr\left\{ \left( (\tN_{11} - \E\{\tN_{11}\})^\sT (\tN_{11} -
    \E\{\tN_{11}\}) \right)^m \right\} &=  \alpha_2^{2m} \sum_{\ell\in
    \frakL(D)}\prod_{e=\{u, v\}\in E_D } g_{\ell(u)\ell(v)}.
  \end{align}
  Let $\frakL_2(D(m))$ denote the set of contributing labelings of $D(m)$. 
  By Lemma \ref{lem:gentrbnd}, it suffices to show that $\max_{\ell\in
  \frakL_2(D(m))}
  \abs{\range(\ell)} \le m+1$.
  Since for a contributing
  labeling $\ell$ of $D(m)$, every edge must occur at least twice, there are at
  most $m$ unique labelings of the edges of $D(m)$. If we consider the graph
  obtained from $(D, \ell)$ by identifying in $D$ the vertices with the same
  label, we obtain a connected graph with at most $m$ edges, hence at most $m+1$
  unique vertices. This implies that there are at most $m+1$ unique labels in
  the range of a contributing labeling $\ell$. Hence with probability at
  least $1-n^{-5}$:
  \begin{align}
    \norm{\tN_{11} - \E\{\tN_{11}\}}_2 &\ls \alpha_2 \barn^{1/2},
  \end{align}
  Hence with the same probability:
  \begin{align}
    \tN_{11} &\mge (\alpha_1 - \alpha_2 p - C\alpha_2 \barn^{1/2})\id_n +
    (\alpha_2 p-\alpha_1^2) \,n\, \qproj_n,
  \end{align}
  for some constant $C$. 
  Under the condition $\alpha_1/2 - \alpha_2 p \gs \alpha_2 \barn^{1/2}$ (with a
  sufficiently large constant which we suppress) we have that:
\begin{align}
  \tN_{11} &\mge \frac{\alpha_1}{2}\id_n + (\alpha_2 p -
  \alpha_1^2) \,n\, \qproj_n\, ,
\end{align}
or, equivalently,
\begin{align} 
  \tN_{11} &\mge \frac{\alpha_1}{2}\qproj_n^\perp + (\alpha_2 p - \alpha_1^2 ) \,n\, \qproj_n.
\end{align}
Inverting this inequality yields the claim for $\tN_{11}^{-1}$. 
This completes the proof of the proposition.
\end{proof}

\subsection{Controlling $\tN_{22} - \E\{\tN_{22}\}$} \label{subsec:deviationM22}

The following proposition is the key result of this subsection. 
\begin{proposition}\label{prop:deviationM22}
  With probability at least $1-25n^{-5}$
  the following hold:
  \begin{align}
    \text{ For } a \in \{0, 1\}\quad\quad \norm{\proj_a (\tN_{22} -
    \E\{\tN_{22}\})\proj_a}_2 &\ls  \alpha_3 \barn^{1/2} +
    \alpha_4 \barn^{3/2}\, ,\label{eq:ClaimN22_1}\\
    \norm{\proj_2 (\tN_{22} - \E\{\tN_{22}\})\proj_2}_2 &\ls \alpha_3
    \barn^{1/2} + \alpha_4 \barn\, , \label{eq:ClaimN22_2}\\ 
    \text{ For } a\ne b \in\{ 0, 1, 2 \}\quad\quad \norm{\proj_a(\tN_{22} -
    \E\{\tN_{22}\})\proj_b}_2 &\ls \alpha_3 \barn^{1/2} +
    \alpha_4 \barn^{3/2}\, . \label{eq:ClaimN22_3}
  \end{align}
\end{proposition}

Recall that:
\begin{align}
  (\tN_{22})_{A,B} &= \begin{cases}
    -\alpha_2^2 + \alpha_2 &\text{if } A =  B \\
    -\alpha_2^2 + \alpha_3 (p + g_{t(A)t(B)}) &\text{if } h(A) = h(B), A\ne B \\
    -\alpha_2^2 + \alpha_3 (p + g_{h(A)t(B)}) &\text{if } t(A) = h(B), A\ne B \\
    -\alpha_2^2 + \alpha_3 (p + g_{t(A)h(B)}) &\text{if } h(A) = t(B), A\ne B \\
    -\alpha_2^2 + \alpha_3 (p + g_{h(A)h(B)}) &\text{if } t(A) = t(B), A\ne B \\
    -\alpha_2^2 + \alpha_4 (p + g_{h(A)h(B)}) (p + g_{h(A)t(B)}) (p +
    g_{t(A)h(B)})(p + g_{t(A)t(B)}) &\text{if $\abs{A\cap B} = 0$.}
  \end{cases}\label{eq:M22AB}
\end{align}

When $\abs{A\cap B} = 0$ (last case above) we can expand $\tN_{A, B}$
as a sum of sixteen terms:
\begin{align}
  \tN_{A, B} &= \alpha_4  (p+g_{h(A)h(B)})(p+g_{h(A)t(B)})(p+
  g_{t(A)h(B)})(p+g_{t(A)t(B)})-\alpha_2^2 \\
  &= (\alpha_4 p^4-\alpha_2^2) + \alpha_4 p^3 (g_{h(A)h(B)} + g_{h(A)t(B)} + g_{t(A)h(B)} +
  g_{t(A)t(B)}) \nonumber \\
 &\quad + \alpha_4 p^2 ( g_{h(A)h(B)}g_{h(A)t(B)}
 +g_{h(A)h(B)}g_{t(A)h(B)} + g_{h(A)h(B)}g_{t(A)t(B)}\nonumber\\
 &\quad\quad\quad+ g_{h(A)t(B)}g_{t(A)h(B)} + g_{h(A)t(B)}g_{t(A)t(B)} +
 g_{t(A)h(B)}g_{t(A)t(B)} ) \nonumber\quad \\
& \quad + \alpha_4 p (g_{h(A)h(B)}g_{h(A)t(B)}g_{t(A)h(B)} +g_{h(A)h(B)}g_{h(A)t(B)}g_{t(A)t(B)} \nonumber\\
&\quad\quad\quad+ g_{h(A)h(B)}g_{t(A)h(B)}g_{t(A) t(B)} + g_{h(A)t(B)}g_{t(A)h(B)}g_{t(A)t(B)}  ) \nonumber \\
 &\quad + \alpha_4 g_{h(A)h(B)}g_{h(A)t(B)}g_{t(A)h(B)}g_{t(A)t(B)}.
\end{align}

Compactly, we can represent the above summation as follows. Each term above 
is indexed by a pair $(\eta, \nu)$ where $0\le \eta\le 4$ denotes
the number of variables $g_{\cdot,\cdot}$ occurring in the product, and $\nu
\le \binom{4}{\eta}$ determines exactly which $\eta$-tuple of $g$
variables occur. For instance, when $\eta = 1$, we have $\binom{4}{1}$ terms
$\alpha_4 p^3 g_{h(A)h(B)}, \alpha_4 p^3 g_{h(A)t(B)}, \alpha_4 p^3
g_{t(A)h(B)}, \alpha_4 p^3g_{t(A)t(B)}$. 
Equivalently, if $R_{A, B}(\eta, \nu)$ 
is a labeled $(\eta, \nu)$-ribbon with exactly one face
and vertices labeled  $h(A), t(A), h(B), t(B)$  in order, each term corresponds to one
specific class and type of ribbon, i.e. 
\begin{align*}
  \tN_{A, B} &= \sum_{\eta, \nu} \alpha_4 p^{4-\eta} \prod_{e=\{i, j\}\in
  R_{A, B}(\eta, \nu)} g_{ij}.
\end{align*}
The exact mapping of the pair $(\eta, \nu)$ to the choice of edges
in $R_{A, B}(\eta, \nu)$ is given in Table \ref{tab:ribbondef}. 
With a slight abuse of terminology, 
we refer to $\eta$ as the
\emph{class} and $\nu$ the \emph{type} of the term. 
We define the matrices $J_{\eta, \nu}$ (for $\eta = 1, 2, 3, 4$ and $\nu =
\binom{4}{\eta}$) and $K$ as follows.
\begin{align}
  (J_{\eta, \nu})_{A, B} &\equiv \begin{cases}
    \alpha_4 p^{4-\eta}\prod_{\{i, j\}\in R_{A, B}(\eta,\nu)} g_{ij} &\text{ if
    } \abs{A\cap B} =0\, , \\
    0 &\text{ otherwise.}
  \end{cases}\label{eq:Jdef} \\
  K_{A, B} &\equiv \begin{cases}
\alpha_3 g_{t(A)t(B)} &\text{if } h(A) = h(B), A\ne B\, ,\\
\alpha_3 g_{h(A)t(B)} &\text{if } t(A) = h(B), A\ne B\, ,\\
\alpha_3 g_{t(A)h(B)}  &\text{if }h(A) = t(B), A\ne B \, ,\\
\alpha_3 g_{h(A)h(B)} &\text{if } t(A) = t(B), A\ne B\, ,\\
0 &\text{ otherwise.}
  \end{cases}\label{eq:Kdef}
\end{align}

The matrices $J_{\eta, \nu}$ vanish on the set of entries
$A, B$ where $A$ and $B$ have non-zero intersection. This causes
the failure of certain useful spectral properties with respect to the spaces
$\bbV_0, \bbV_1, \bbV_2$. Consequently, for our proof, it is useful to
define the matrices $\tJ_{\eta, \nu}$ that do not have this constraint. 
\begin{align} 
  \label{eq:tJdef}
  (\tJ_{\eta, \nu})_{A, B} &\equiv 
    \alpha_4 p^{4-\eta}\prod_{\{i, j\}\in R_{A, B}(\eta,\nu)} g_{ij}. 
\end{align}
Here we ignore the constraint that $A, B$ do not intersect, and follow
the convention that $g_{ii} =0$ for every $i\in[n]$. 

Thus, with \myeqref{eq:M22AB} we arrive at the following expansion:
\begin{align}
  \tN_{22}- \E\{\tN_{22}\} &= K + \sum_{\eta=1}^{4} \sum_{\nu=1}^{\binom{4}{\eta}}
  J_{\eta, \nu} \\
  &= K + J_{2, 1} + J_{2, 6} + J_{4, 1} + \sum_{\nu=1}^{4} J_{3, \nu}+
    \sum_{\nu =1}^{4} (J_{1, \nu} -
  \tJ_{1, \nu}) + \sum_{ \nu=2}^{5} (J_{2, \nu} - \tJ_{2, \nu})\nonumber \\
  &\quad+ \sum_{\nu=1}^{4}\tJ_{1, \nu} + \sum_{\nu=2}^{5} \tJ_{2, \nu}. \label{eq:M22expansion}
\end{align}

We now prove a sequence of lemmas regarding the spectral
properties of the matrices $K, J_{\eta, \nu}$.
The first one concerns the case
$\eta =2$, $\nu=1, 6$ and $\eta=4$, $\nu=1$. 
\begin{lemma}
  With probability at least $1- 3n^{-5}$, we have that:
  \begin{align}
    \norm{J_{2, 1} + J_{2, 6} + J_{4, 1}   }_2 
    &\ls \alpha_4 \barn
  \end{align}
  \label{lem:normbndeta24}
\end{lemma}
\begin{proof}
   By the triangle inequality:
  \begin{align}
    \norm{J_{2, 1} + J_{2, 6} + J_{4, 1}}_2  &\le \norm{J_{2, 1}}_2 + \norm{J_{2,
    6}}_2 +  \norm{J_{4, 1}}_2. 
    \end{align}
    We prove that with probability at least $1-n^{-5}$
    \begin{align}
      \norm{J_{\eta, \nu}}_2 &\ls \alpha_4 \barn,
    \end{align}
    for  $(\eta, \nu)=(2, 1), (2, 6), (4, 1)$. The claim then
    follows by a union bound. 

  Let $R(\eta, \nu, m)$ denote a
  $(\eta, \nu)$-ribbon of length $2m$. Then, by expanding the product we have:
  \begin{align}
    \Tr\left\{(J_{\eta, \nu}^\sT J_{\eta, \nu})^{r}\right\} &= 
    \sum_{\ell\in \frakL(R(\eta, \nu,  m)) } (\alpha_4
    p^{4-\eta})^{2m}\left\{\prod_{e\in R(\eta,
    \nu, m)} g_{\ell(e)} \right\}.
  \end{align}
  Here we write $\ell(e)$ in place of the pair $\ell(u), \ell(v)$ when $u, v$
  are the end vertices of $e$. Since $R(\eta, \nu, m)$ has $4m+2$
  vertices, by Lemma \ref{lem:gentrbnd} it
  suffices to prove that $\max_{\ell\in \frakL_2(R(\eta,
  \nu, m))}\range(\ell) =2m+2$. 

  We first prove this for the case $\eta=2$ and $\nu=1, 6$. Let
  $\ell$ be a contributing labeling of the ribbon $R(\eta, \nu, m)$ of length
  $2m$. Let $\bG({\eta, \nu})$ denote the graph obtained by identifying
  in $R(\eta, \nu, m)$ every vertex with the same label according to $\ell$. 
  We have:
  \begin{align}
    \text{ \# connected components in $\bG({\eta,\nu})$} &\le  \text{ \#
      connected components in $R({\eta,\nu}, m)$} = 2\\
      \text{ \# edges in $\bG_{\eta, \nu}$} &\le \frac{\text{ \# edges in
	$R({\eta, \nu}, m)$}}{2} = 2m.
  \end{align}
  It follows that there are at most $2m+2$ unique vertices in $\bG({\eta, \nu},
  m)$ 
  and hence, at most $2m+2$ unique labels in $\range(\ell)$.

  We now prove the condition  $\max_{\ell\in \frakL_2(R({\eta,
  \nu}, m))}\range(\ell) =2m+2$ for $\eta=4, \nu=1$, induction on $m$. 
The base case is $m=1$ (or a ribbon of length 2), wherein it is obvious that a contributing labeling
$\ell$ can have at most $4 = 2m+2$ unique labels.
Now, assume the claim is true for ribbons of length at most $2m>1$ and we 
will prove it for $R({4, 1}, m+1)$ of length $2m+2$. Consider any contributing
labeling $\ell$ of $R({4, 1}, m+1)$. We now have the following  cases
\begin{enumerate}
  \item For every vertex $u\in R({4, 1}, m+1)$, there exists $u'\ne u$ such that
    $\ell(u')=\ell(u)$.
  \item There exists vertex $u\in R({4, 1}, m+1)$ with a unique label
    $i=\ell(u)$ and the degree of $u$ is 4.
\end{enumerate}
For case 1, if every label in the range of $\ell$
occurs at least twice in $R({4, 1}, m)$, the number of unique labels is
bounded by $2(m+1)$, since $R({4, 1}, m)$
has only $4(m+1) $ vertices, hence the claim follows.

For case 2, 
let $(u_1, v_1)$ and $(u_2, v_2)$
be the neighboring couples of $u$. If $u$ is connected to all of $u_1, v_1, u_2,
v_2$, since the edges connected to $u$ must occur twice, it must hold that
$\ell(u_1) = \ell(u_2)$ and $\ell(v_1)=\ell(v_2)$ (recall indeed that
$\ell(u_1)<\ell(v_1)$, $\ell(u_2)<\ell(v_2)$ by definition of a valid labeling). Hence, we can contract the
ribbon removing the couple containing $u$ and all edges and identifying the
couples $(u_1, v_1)$ with $(u_2, v_2)$. We obtain now a ribbon
$\tilde{R}{4, 1}, m)$ of length $2m$ and an induced labeling 
$\tilde{\ell}$ thereof which is contributing. By induction hypothesis, 
$\range(\tilde{\ell})\le 2m+2$,  hence 
$\range(\ell) = \range(\tilde{\ell}) + 2\le  2(m+1)+2$. This completes the
proof. 
\end{proof}

\begin{lemma}\label{lem:normbndeta3}
  With probability at least $1- 8n^{-5}$, we have 
  \begin{align}
    \norm{ \sum_{\nu=1}^{4} J_{3, \nu}  }_2 
    &\ls \alpha_4 p  \barn\, .
  \end{align}
\end{lemma}
\begin{proof}
  By the triangle inequality, it suffices to show that for $\nu\in\{1,\dots,4\}$, with 
  probability $1-n^{-(\Gamma-2)/2}$:
  \begin{align}
    \norm{J_{3, \nu}}_2 &\le \alpha_4 p \barn\, .
  \end{align}
  
  We prove the above for the case $\nu = 2$. The other case follow 
  from analogous arguments. Firstly, define the matrices $\tJ_{3, 2}\in
  \reals^{\nCe \times \nCe}$
  and $Q \in \reals^{n^2 \times n^2}$ as follows:
  \begin{align}
    (\tJ_{3, 2})_{\{i, j\}, \{k, l\}} &= 
      \alpha_4 p g_{ik}g_{il}g_{jl}, \\
    Q_{(i, j), (k, l)} &= g_{ik}g_{il}g_{jl}.
  \end{align}
  Note also that $\tJ_{3, 2}$ differs from $J_{3,2}$ only
  in the entries $\{i, j\}, \{k, \ell\}$ where $j=k$. 
  The rows (columns) of $Q$ above are indexed by \emph{ordered} pairs
  $(i, j)\in[n]\times[n]$.
  Now we define the projector $\proj_{\nCe} :
  \reals^{n^2}\to\reals^{\nCe}$ by letting,
for all $i,j\in [n]$,
  \begin{align}
    (\proj_{\nCe}(x))_{\{i, j\}} &= x_{(i,j)}\, .
  \end{align}
  Then we have $\tJ_{3, 2} = \alpha_4 p\proj_{\nCe} Q \proj_{\nCe}^{\sT}$ and, consequently,
   $\|\tJ_{3, 2}\|_2 \le \alpha_4 p \norm{Q}_2$. Therefore it  suffices
  to bound the latter, which we do again by the moment method. 
  Firstly we define:
\begin{align}
U_{(i,j),(k,l)} &= \sum_{q\in [n]}g_{iq}g_{qk}g_{ij}\ind(j=l) \, ,\\
D_{(i,j),(k,l)} &= \sum_{q\in [n]}g_{jq}g_{ql}g_{ij}\ind(i=k) \,.
\end{align}
  
Then we have, for any integer $m\ge 1$, 
\begin{align*}
&  \Tr((Q^\sT Q)^{m})  
= \sum_{i_1,i_2,\dots,i_{2m}\in[n]}
  \sum_{j_1,j_2,\dots,j_{m}\in[n]}
  Q^{\sT}_{(i_1,j_1),(i_2,j_2)}Q_{(i_2,j_2),(i_3,j_3)}Q^{\sT}_{(i_3,j_3),(i_4,j_4)}\cdots
Q_{(i_{2m},j_{2m}),(i_1,j_1)}\\
&= \sum_{i_1,i_2,\dots,i_{2m}\in[n]}
  \sum_{j_1,j_2,\dots,j_{m}\in[n]} 
\big(g_{i_1i_2}g_{j_1j_2}g_{i_1j_2}\big)\cdot
\big(g_{i_2i_3}g_{j_2j_3}g_{j_2i_3}\big)\cdot
\big(g_{i_3i_4}g_{j_3j_4}g_{i_3j_4}\big)\cdots
\big(g_{i_{2m}i_1}g_{j_{2m}j_1}g_{j_{2m}i_1}\big)\\
& = \sum_{i_1,i_2,\dots,i_{2m}\in[n]}
  \sum_{j_1,j_2,\dots,j_{2m}\in[n]} 
  \big(g_{i_1i_2}g_{i_2i_3}g_{i_3i_4}\cdots g_{i_{2m}i_1}\big)
  \big(g_{j_1j_2}g_{j_2j_3}g_{j_3j_4}\cdots g_{j_{2m}j_1}\big)
  \big(g_{i_1j_2}g_{j_2i_3}g_{i_3j_4}\cdots g_{j_{2m}i_1}\big)\\
& = \sum_{i_1,i_2,\dots,i_{2m}\in[n]}
  \sum_{j_1,j_2,\dots,j_{2m}\in[n]} 
\big(g_{i_1i_2}g_{i_2i_3}g_{i_1j_2}\big)
\big(g_{j_2j_3}g_{j_3j_4}g_{j_2i_3}\big)
\big(g_{i_3i_4}g_{i_4i_5}g_{i_3j_4}\big)
\cdots
\big(g_{j_{2m}j_1}g_{j_1j_2}g_{j_{2m}i_1}\big)\,.
\end{align*}
Then we have
\begin{align}
\Tr( (Q^\sT Q)^{m})  &= \Tr((UD)^{m})\, .
\end{align}
Hence
\begin{align}
\|Q\|_2\le \Tr( (Q^\sT Q)^{m})^{1/2m}  \le \Tr((UD)^{m})^{1/2m}
\le \big(n^2\|U\|_2^{m}\|D\|_2^{m}\big)^{1/2m}
\le n^{1/m} \|U\|_2
\, ,
\end{align}
where in the last step we used the fact that $\|U\|_2=\|D\|_2$ by
symmetry. Since $m$ can be taken arbitrarily large, we conclude that
$\|Q\|_2\le\|U\|_2$ and we proceed to bound the latter.

Now let $T\in\reals^{n^2\times n^2}$ be the element-wise multiplication
by $g$, i.e.
\begin{align}
T_{(i,j),(k,l)} = g_{ij}\ind(i=k)\ind(j=l)\, .
\end{align}
Then we have
\begin{align}
U = T\cdot \big(g^2\otimes \id_n\big)
\end{align}
Here $g \in \reals^{n\times n}$ is the matrix with $i, j$ entry
being $g_{ij}$. 
Since $|g_{ij}|\le 1$, we have $\|T\|_2\le 1$ and therefore
\begin{align}
\|Q\|_2\le \|U\|_2\le \|T\|_2 \|g^2\otimes \id\|_2\le
\|g^2\otimes \id\|_2\le \|g^2\|_2\le \|g\|_2^2\,.
\end{align}
Finally, similar to Proposition \ref{prop:deviationM11} we 
have that $\norm{g} \ls \barn^{1/2}$ with probability
at least $1-n^{-5}$, hence with the same probability:
\begin{align}
  \norm{\tJ_{3, 2}}_2 &\ls \alpha_4 p \barn\, .\label{eq:tJ32}
\end{align}
By triangle inequality $\norm{J_{3,2}}_2 \le \|\tJ_{3, 2}\|_2 + \|\tJ_{3,
2}-J_{3, 2}\|_2$, 
hence to complete the proof we now bound $\|\tJ_{3, 2} - J_{3, 2}\|_2$ using
the moment method. Recall that $\tJ_{3, 2}$ and $J_{3, 2}$ differ in the entry $\{i, j\}, \{k, \ell\}$
only if $j = k$. Hence: 
\begin{align}
  \Tr\left\{\big( (\tJ_{3, 2} - J_{3, 2})^\sT (\tJ_{3, 2} - J_{3, 2})\big)^{m}\right\} 
  &= (\alpha_4 p)^{2m} \sum_{i_1\dots i_{2m}, j_1 \dots j_{2m}, \forall q\, i_q < j_q } 
  \prod_{q=1}^{m}\bigg(g_{i_q i_{q+1}}g_{i_q j_{q+1}}g_{j_q j_{q+1}}
 \nonumber\\  &\quad\quad\quad g_{i_{q+1} i_{q+2}} g_{j_{q+1}i_{q+2}}g_{j_{q+1} j_{q+2}}
  \ind(j_1 = i_2 = j_3 = i_4 =\dots= i_{2m})\bigg) \\
  &= (\alpha_4 p)^{2m} \sum_{\ell\in \widetilde{\frakL}(R(3, 2, m)} \prod_{e=\{u, v\}\in R(3, 2, m)} g_{\ell(u)\ell(v)}.
\end{align}
Here, $R(3, 2, m)$ is a $(3, 2)$-ribbon of length $2m$ and $\widetilde{\frakL}(R(3, 2, m)$ is a 
collection of labelings of $R(3, 2, m)$ satisfying the following criteria
\begin{enumerate}
  \item For every couple $(u, v) \in R(3, 2, m)$, $\ell(u) < \ell(v)$.
  \item Let $(u_1, v_1), (u_2, v_2) \dots (u_{2m}, v_{2m})$ denote the couples
    in $R(3, 2, m)$. Then $\ell(v_1) = \ell(u_2) = \ell(v_3) = \ell(u_4) \dots$.
\end{enumerate}
Let $\widetilde{\frakL}_2(R(3, 2, m))$ denote the subset of contributing labelings, 
i.e. those that
satisfy the additional criterion that every labeled edge is repeated
twice. 
By Lemma \ref{lem:gentrbnd} it suffices to show that 
$v_*(R(3, 2, m)) = \max_{\ell\in \widetilde{\frakL}(R(3, 2, m))}\abs{\range(\ell)} \le m+2$. 
We prove this by induction. For the base case of $m=1$, since
every edge is repeated twice under a contributing labeling, it
is easy to see that there are at most $3$ unique labels. Assume
the induction hypothesis that
$v_*{R(3, 2, m-1)}\le m+1$. 
Let $\ell$ be a contributing labeling of $R(3, 2, m)$. Then one of the following
must happen:
\begin{enumerate}
  \item No vertex in $R(3, 2, m)$ has a unique label under $\ell$. 
  \item There exists a vertex $w$ of degree 4 with a unique label under $\ell$.
\end{enumerate}
The second condition follows because the vertices of degree smaller than 4 
already have non-unique labels due to condition 2 of the labeling set
$\widetilde{\frakL}(R(3, 2, m))$. 

In case 1, $R(3, 2, m)$ can have at most
$2m/2 + 1 = m+1 < m+2$ unique labels under $\ell$. 
In case 2, since $w$ has a unique label and degree 4
the neighboring $(u, v), (u', v')$ have the same labels
under $\ell$ i.e. $\ell(u)=\ell(u')$ and $\ell(v)=\ell(v')$. Hence
we can identify the couples $(u, v), (u', v')$,  delete $w$ and
its incident edges to obtain a ribbon $\tilde{R}(3, 2, m-1)$ of
length $2m-2$ and an induced labeling $\tilde{\ell}$ thereof. 
By the induction hypothesis $\range(\tilde{\ell})\le m+1$ hence
$\range(\tilde{\ell}) = \range(\tilde{\ell})+1 \le m+2$, as 
required. By Lemma \ref{lem:gentrbnd} we obtain that
$\norm{\tJ_{3, 2} - J_{3, 2}}_2 \ls \alpha_4 p \barn$ with
probability at least $1-n^{-5}$.

By Eq.~(\ref{eq:tJ32}), it follows that with probability at least $1-2n^{-5}$,
$\|J_{3, 2}\|_2 \ls \alpha_4 p \barn \ls \alpha_4 \barn$. This completes the proof of the lemma.
\end{proof}

For the case $\eta=1$ we prove the following
\begin{lemma}
  Recall that $\proj_2:\reals^{\nCe}\to \reals^{\nCe}$ is the
  orthogonal projector onto the space $\bbV_2\subseteq
  \reals^{\nCe}$ (defined in Section \ref{subsec:expectation}). 
  Firstly, we have that $\proj_2(\sum_{\nu=1}^4 \tJ_{1, \nu})\proj_2 = 0$
  Further, with probability at least $1-4n^{-5}$, we have that:
  \begin{align}
    \norm{\sum_{\nu=1}^{4} \tJ_{1, \nu}}_2 &\ls \alpha_4 \barn^{3/2}\label{eq:tJ1}
  \end{align}
  \label{lem:normbndeta1}
\end{lemma}
\begin{proof}
  Recall from the definition of $\tJ_{1, \nu}$ that
  \begin{align}
    \sum_{\nu=1}^{4}(\tJ_{1, \nu})_{\{i, j\}, \{k, \ell\}} &= p^{3}(g_{ik} + g_{i\ell} + g_{jk}+ g_{j\ell}).
  \end{align}
  Now, for any $v\in \reals^{\nCe}$:
  \begin{align*}
    \left( \sum_{\nu=1}^4\tJ_{1, \nu}v \right)_{\{i,j\}} &= \sum_{k<\ell} p^3(g_{ik} + g_{i\ell} + g_{jk} + g_{j\ell})v_{\{k,\ell\}}\\
    &= u_i + u_j,
  \end{align*}
  where we define $u_i\equiv \sum_{k\le \ell} p^3( g_{ik} + g_{i\ell})v_{\{k,\ell\}}$. 
  It follows that $\sum_{\nu=1}^4\tJ_{\eta, \nu}v \in \bbV_2^\perp = \bbV_0\oplus\bbV_1$, and
  hence $\proj_2\sum_{\nu=1}^4\tJ_{\eta, \nu} = 0$. Since $\sum_{\nu=1}^4\tJ_{1, \nu}$ is symmetric
  we obtain the first claim.

  We prove the second claim --cf. Eq. (\ref{eq:tJ1})-- by the moment method, 
  similar to Lemma \ref{lem:normbndeta24}. 
  Let $R({1, \nu}, m)$ be a $(1, \nu)$-ribbon of length $2m$. Then:
  \begin{align}
    \Tr\left\{(\tJ_{1, \nu}^\sT \tJ_{1, \nu})^{r}\right\} &= 
   (\alpha_4 p^3)   \sum_{\ell\in \frakL(R({1, \nu}, m)) } \left\{
      \prod_{e=\{u, v\}\in R({1,
    \nu}, m)} g_{\ell(u)\ell(v)} \right\}.
  \end{align}
  By Lemma \ref{lem:gentrbnd} it suffices to prove that
  $v_*(R({1, \nu}, m)) = 3m+2$. The claim then follows, using Lemma
  \ref{lem:gentrbnd} and the union bound.

  Let $\ell\in \frakL_2(R({1, \nu}, m))$ be a contributing labeling
  of a ribbon $R({1, \nu}, m)$ of  length $2m$. Let $\bG({1, \nu}, m)$ be
  the graph obtained by identifying vertices in $R({1, \nu}, m)$ with 
  the same label. Notice that $R({1, \nu}, m)$ is a union of a cycle
  $D(m)$ of length $2m$ and $2m+1$ isolated vertices. The isolated vertices
  can have arbitrary labels, hence $v_*(R({1, \nu}, m)) = 2m+1 + v_*(D(2m)) =
  3m+2$ as proved in Proposition \ref{prop:deviationM11}.
\end{proof}

In a similar fashion, we bound the norm of the terms $\tJ_{2, 2}$,  $\tJ_{2, 3}$, 
$\tJ_{2, 4}$,   $\tJ_{2, 5}$:

\begin{lemma} 
  \label{lem:normbndeta2}
  We have that:
  \begin{align}
    (\tJ_{2, 2} + \tJ_{2, 4})\proj_2 &= 0,\label{eq:projeta21} \\
    \proj_2(\tJ_{2, 3} + \tJ_{2, 5}) &= 0.\label{eq:projeta22}
  \end{align}
  Further with probability at least $1- 2n^{-4}$
  \begin{align}
    \norm{\tJ_{2, 2}}_2 &\ls (\alpha_4 p^2)  \barn^{3/2},\\
    \norm{\tJ_{2, 4}}_2 &\ls (\alpha_4 p^2)  \barn^{3/2}.
  \end{align}
\end{lemma}
\begin{proof}
  It is easy to check that $\tJ_{2, 2} = \tJ_{2, 3}^\sT$ 
  and $\tJ_{2, 4} = \tJ_{2, 5}^{\sT}$. 
  We prove \myeqref{eq:projeta22}, from which 
  \myeqref{eq:projeta21} follows by taking transposes
  of each side. From the definition of $\tJ_{2, \nu}$ 
  we have for any $v\in \reals^{\nCe}$
  \begin{align}
    (\tJ_{2, 3}v + \tJ_{2, 5}v)_{\{i, j\}} &= \sum_{k<\ell} p^2(g_{ik}g_{i\ell} + g_{jk}g_{j\ell})v_{\{k, \ell\}}\\
    &= u_i + u_j,
  \end{align}
  where we let $u_i \equiv = \sum_{k<\ell}p^2(g_{ik}g_{i\ell})v_{\{k, \ell\}}$. It follows that
  $(\tJ_{2, 3}v + \tJ_{2, 5})v \in \bbV_0\oplus\bbV_1$ hence $\proj_2 (\tJ_{2, 3} + \tJ_{2, 5}) = 0$.

  We prove the claim on the spectral norm for $\tJ_{2, 2}$. The claim for $\tJ_{2, 4}$ holds in an 
  analogous fashion. Let $R({2, 2}, m)$ be a $(2,2)$-ribbon of length $m$.
  Then:
  \begin{align}
    \Tr\left\{(\tJ_{2, 2}^\sT \tJ_{2, 2})^{m}\right\} &= 
    \sum_{\ell\in \frakL(R({2, 2}, m)) } (\alpha_4 p^{2})^{2m}\prod_{e=\{u,
    v\}\in R({2,
    2}, m))} g_{\ell(u)\ell(v)}.
  \end{align}
  By Lemma \ref{lem:gentrbnd}, it suffices to show that
  $v_*(R({2, 2}, m)) = 3m+2$. 
i.e a contributing labeling $\ell$ maps to at
most $3m + 2$ unique labels. Notice that $R({2, 2}, m)$ is the union 
of $m+1$ isolated vertices and a bridge $B(m)$ of length $2m$. The isolated vertices
are unconstrained and hence contribute at most $m+1$ new labels. It suffices,
hence, to prove that $B(m)$ has at most $2m +1$ unique labels under its labeling
$\ell_{B(m)}$ induced by $\ell$. Since, $\ell_{B(m)}$ is contributing for
$B(m)$, it suffices that $v_*(B(m))= 2m+1$.   We prove this by induction on $m$. 
In the base case of $m=1$, this implies it has at most
$3=(2\cdot 1+1)$ unique labels. Assuming that the claim is true for bridges
of length at most  $2m$ for $m>1$, we show
that it holds for a bridge $B(m+1)$ of length $2m+2$. $B(m+1)$ contains $3m+4$ vertices
hence there are 3 cases:
\begin{enumerate}
  \item For every vertex $u\in B$ there exists a different vertex $u'\in B$ such
    that $\ell_{B}(u) = \ell_{B}(u')$. 
  \item There exists a vertex $u\in B$ which has a unique label under
    $\ell_{B}$ and $u$ has degree 4. 
  \item There exists a vertex $u\in B$ which has a unique label under $\ell_B$
    with degree 2. 
\end{enumerate}
In the first case, $\abs{\range(\ell)}\le (3m+4))/2\le 2(m+1)+1$ hence the claim
holds. 

In the second case, we have that
the neighboring couples are $(u_1, v_1)$, $(u_2, v_2)$ then
$\ell_{B(m+1)}(u_1)=\ell_{B(m+1)}(u_2)$ and
$\ell_{B(m+1)}(v_1)=\ell_{B(m+1)}(v_2)$. We can then contract the neighbors of $u$ and delete
$u$ and incident edges to obtain a bridge $\tilde{B}(m)$ (and induced labeling
$\ell_{\tilde{B}(m)}$ of length $2m$). By induction $\ell_{\tilde{B}(m)}$ maps to at
most $2m+1$ labels, hence $\ell_{B(m)}$ to at most $2m+1+1 \le 2(m+1) +1$ labels. 

In the third case, if $u$ has neighbors $u_1,u_2$ then
$\ell_{B(m+1)}(u_1)=\ell_{B(m+1)}(u_2)$. 
If we now identify the neighbors of $u$ with the same
label, and delete $u$ and the edges incident on it, we obtain a bridge
$\tilde{B}(m)$ of length $2m$, and an induced labeling
$\ell_{\tilde{B}(m)}$ which is
contributing. By induction, $\tilde{B}(m)$ has at most $2m+1$ unique labels, 
hence $B(m+1)$ has at most $2m+1+2 = 2(m+1) + 1$ unique labels. This
completes the induction.
\end{proof}

Finally, we have to deal with the remainder terms (recall that matrix
$K$ is defined in Eq.~(\ref{eq:Kdef})).
\begin{lemma} 
  \label{lem:normbndK}
  We have with probability at least $1-n^{-5}$ that:
  \begin{align}
    \norm{K}_2 &\ls \alpha_3 \barn^{1/2}
  \end{align}
\end{lemma}
\begin{proof}
  We compute $\Tr\left\{
    (K^\sT K)^{m}
  \right\}$. Note that:
  \begin{align}
    \Tr\left\{ (K^{\sT}K)^m \right\} &= \sum_{A_1, B_1 \dots A_{m}
    B_{m}} \prod_{l=1}^r(K_{A_l B_{l}} K_{A_{l+1}B_{l}})\\
    &= \sum_{A_1, B_1 \dots A_m B_m} \prod_{l=1}^rK_{A_l B_{l}} K_{A_{l+1}B_{l}}
    \ind(\abs{A_l\cap B_l} = 1)\ind(\abs{A_{l+1}\cap B_l}=1) .
  \end{align}
  Here we set $A_{m+1}\equiv A_1$. The second equality follows since $K$ is
  supported on entries $A, B$ such that $A, B$ share exactly one vertex. 
  Recalling the definition of 
  star ribbons, each term that does not vanish in the summation above
  corresponds a labeling of a star ribbon $S({2, 1}, m)\in \cS_{2, 1}^{m}$ formed from a $(2,
  1)$-ribbon of length $2m$, i.e. we have:
  \begin{align}
    \Tr\left\{ (K^{\sT}K)^m \right\} &= \alpha_3^{2m}\sum_{S({2, 1}, m)\in \cS_{2, 1}^m}\sum_{\ell\in \frakL(S({2, 1}, m))}
    \prod_{e=\{u, v\}\in S({2, 1}, m)} g_{\ell(u), \ell(v)}.
  \end{align}
  Since there are at most $2^{2m} = 4^m$ star ribbons of length $2m$, it suffices by a simple extension
  of 
  Lemma \ref{lem:gentrbnd}, to show that $v_*(S({2, 1}, m)) =
  m+2$. Note that every $S({2, 1}, m)$ is a union of 2 paths, one
  of length $m'$ and the other of length $2m-m'$ for some $m' \in [2m]$,
  hence has at most $2$ connected components. Let $\ell$ be a contributing
  labeling of $S({2, 1}, m)$ and $\bG_{S({2, 1}, m)}$ be the graph obtained
  by identifying vertices in $S({2, 1}, m)$ with the same label. Since $S({2,
  1}, m)$
  is a union of two paths, $\bG_{S({2, 1}, m)}$ has at most 2 connected components.
  Furthermore, since $\ell$ is a contributing labeling, every labeled edge in $S({2, 1}, m)$
  repeats at least twice, hence $\bG_{2, 1}(m)$ has at most $2m/2 = m$ edges. 
  Consequently, it has at most $m+2$ vertices, implying that $v_*(S({2, 1}, m)) \le m+2$. 
\end{proof}

Finally, we deal with the differences $J_{\eta, \nu} - \tJ_{\eta,
  \nu}$. 
(Recall that $J_{\eta, \nu}$ and $\tJ_{\eta, \nu}$ are defined in
Eqs. (\ref{eq:Jdef}) and  (\ref{eq:tJdef}).)
\begin{lemma} \label{lem:normbndJmtJ}
  With probability at least $1- 6n^{-5}$, 
  for each $\eta\le 2$ and $ \nu \le \binom{4}{\eta}$:
  \begin{align}
    \norm{J_{\eta, \nu} - \tJ_{\eta, \nu}}_2 &\ls \alpha_4  \barn
  \end{align}
\end{lemma}
\begin{proof}
  We first consider $\Tr\left\{ ((\tJ_{\eta, \nu} - J_{\eta, \nu})^\sT
  (\tJ_{\eta, \nu} - J_{\eta, \nu}))^{m} \right\}$. Let $R({\eta,
  \nu}, m)$ be a
  $(\eta, \nu)$-ribbon of length $2m$. As in the previous lemmas, we can write
  $\Tr\left\{( (\tJ_{\eta, \nu} - J_{\eta, \nu})^\sT
  (\tJ_{\eta, \nu} - J_{\eta, \nu}) )^m\right\}$ as a sum over labelings of
  $R({\eta, \nu}, m)$ as follows:
  \begin{align}
    \Tr\left\{ ((\tJ_{\eta, \nu} - J_{\eta, \nu})^\sT
    (\tJ_{\eta, \nu} - J_{\eta, \nu}))^{m} \right\} &= (\alpha_4
    p^{4-\eta})^{2m}\sum_{\ell\in \widetilde{\frakL}(R({\eta, \nu}, m))}
    \prod_{e=\{u, v\}\in R({\eta, \nu}, m)} g_{\ell(u), \ell(v)}.
  \end{align}
  Here we restrict the labelings $\ell$ to the subset
  $\widetilde{\frakL}(R({\eta, \nu}, m)$
  that satisfy the criteria:
  \begin{enumerate}
	\item For every couple $(u, v)$, $\ell(u)<\ell(v)$.
	\item Consider any adjacent pair of couples $(u_1, v_1), (u_2, v_2)$ in
	  $R({\eta, \nu}, m)$,
	  at least one of $u_1, v_1, u_2, v_2$ has degree 0. Assume this is $u_1$ (without loss of generality), 
	  then either $\ell(u_1) = \ell(u_2)$ or $\ell(u_1) = \ell(v_2)$. 
  \end{enumerate}
  On taking expectations the only labelings that do not vanish satisfy the additional 
  criterion that every labeled edge is repeated at least twice in $R({\eta,
  \nu}, m)$. We 
  call this set of labelings $\widetilde{\frakL}_2(R({\eta, \nu}, m))$. As in Lemma \ref{lem:normbndeta2}
  it suffices to show that $\abs{\widetilde{\frakL}_2(R({\eta, \nu}, m))} \le \binom{n}{2m+2}(2^{2m} (2m+2)^{3m+2})$. 
  This follows from the same arguments as in Lemmas \ref{lem:normbndeta2},
  \ref{lem:normbndeta1}
  (for $\eta=1, 2$ respectively), with the additional caveat that the isolated vertices in 
  $R({\eta, \nu}, m)$ are not unconstrained as before. Indeed, once the labels of the connected component
  of $R({\eta, \nu}, m)$ are decided, there are only $2^m$ possible ways of choosing the labels
  for the isolated vertices. Consequently, we have the bound:
  \begin{align}
\E\Tr\left\{ ((\tJ_{\eta, \nu} - J_{\eta, \nu})^\sT
(\tJ_{\eta, \nu} - J_{\eta, \nu}))^m \right\} &\le (\alpha_4
p^{4-\eta})^{2m}\abs{\widetilde{\frakL}_2(R({\eta, \nu}, m))}\\
&\le \binom{n}{2m+2} (2\alpha_4 p^{4-\eta})^{2m}(2m+2)^{3m+2}.
  \end{align}
  Applying Lemma \ref{lem:trpow}, union bound and the triangle inequality yields the final result.
\end{proof}

We can now prove Proposition \ref{prop:deviationM22}. 
\begin{proof}[Proof of Proposition \ref{prop:deviationM22}]
  The intersection of high probability events of Lemmas \ref{lem:normbndeta24}, \ref{lem:normbndeta3}, \ref{lem:normbndeta1}, 
  \ref{lem:normbndeta2}, \ref{lem:normbndK}
  and \ref{lem:normbndJmtJ} holds with probability at least $1-25n^{-5}$. 
  We will condition on this event for the proof of the proposition. 

  We  bound each of the projections $\proj_a (\tN_{22} -
  \E\{\tN_{22}\})\proj_b$ for $a, b \in\{ 0, 1, 2\}$ using the
  decomposition (\ref{eq:M22expansion}).
\begin{itemize}
\item Let us first consider $a=b$, $a,b\in \{0, 1\}$,
  cf. Eq.~(\ref{eq:ClaimN22_1}).
 By application 
  of above lemmas, triangle inequality,  the fact that
  $\norm{\proj_a X \proj_b}_2 \le \norm{\proj_a}_2\norm{X}_2\norm{\proj_b}_2 \le
  \norm{X}_2$ for any $X \in \reals^{\nCe}$ in the decomposition 
  \myeqref{eq:M22expansion}, we get 
  \begin{align}
    \norm{\proj_a(\tN_{22} - \E\{\tN_{22}\})\proj_a}_2 &\ls \alpha_3
    \barn^{1/2} + \alpha_4\bar + \alpha_4\barn^{3/2} \big)\\
    &\ls \alpha_3 \barn^{1/2} +  \alpha_4\barn^{3/2},
  \end{align}
This proves Eq. ~(\ref{eq:ClaimN22_1}).
\item The case $a=b=2$ is treated in the same manner,  with the only
  difference that, when bounding $\norm{\proj_2 (\tN_{22}-
  \E\{\tN_{22}\})\proj_2}$, the terms of the type $\alpha_4
  \barn^{3/2}$ do not appear (see Lemmas
  \ref{lem:normbndeta1}, \ref{lem:normbndeta2}). Hence:
  \begin{align}
    \norm{\proj_2(\tN_{22} - \E\{\tN_{22}\})\proj_2}_2 &\ls \alpha_3
    \barn^{1/2} + \alpha_4\barn \\
    &\le \alpha_3 \barn^{1/2} +  \alpha_4\barn.
  \end{align}
This proves Eq. ~(\ref{eq:ClaimN22_2}).
\item  The bound for the cross terms $\norm{\proj_a (\tN_{22} -
  \E\{\tN_{22}\})\proj_b}_2$ for $a\ne b$ is identical to that for 
  the case $a=b=0$ above. 

This proves Eq. ~(\ref{eq:ClaimN22_3}) and hence finishes our proof 
of Proposition \ref{prop:deviationM22}.
\end{itemize}
\end{proof}

\subsection{Controlling $\tN_{12} - \E\{\tN_{12}\}$} \label{subsec:deviationM12}

We prove the following proposition for the deviation
$\tN_{12} - \E\{\tN_{12}\}$
\begin{proposition}\label{prop:deviationM12}
  With probability at least $1- 5n^{-5}$ the following
  are true. 
  \begin{align}
    \norm{\tN_{12} - \E\{\tN_{12}\}}_2 & \ls  \alpha_3 \barn.
  \end{align}
\end{proposition}

Recall that an entry of $\tN_{12}\in\reals^{\binom{[n]}{1}\times \nCe}$ can be written as:
\begin{align}
  (\tN_{12})_{A, B} &= \begin{cases}
    \alpha_2 - \alpha_1 \alpha_2 &\text{ if } \abs{A\cap B} =1 \\
    \alpha_3 (p + g_{A, h(B)})(p + g_{A, t(B)}) - \alpha_1\alpha_2
    &\text{ otherwise.}
  \end{cases}
\end{align}
Define the matrices $L_{\eta, \nu} \in\reals^{\binom{[n]}{1} \times \nCe}$ for $\eta = 1, 2$, $\nu \le
\binom{\eta}{\nu}$ and $\tL_{1, \nu}$ for $\nu = 1, 2$:
\begin{align}
  (L_{2, 1})_{A, B} &\equiv \begin{cases}
    \alpha_3 g_{A, h(B)}g_{A, t(B)}  &\text{ if } \abs{A \cap B} = 0\\
    0 &\text{ otherwise.}
  \end{cases} \\
  (L_{1, 1})_{A, B} &\equiv \begin{cases}
    \alpha_3 p g_{A, h(B)} &\text{ if } \abs{A \cap B} = 0 \\
    0&\text{ otherwise.}
  \end{cases}\\
  (L_{1, 2})_{A, B} &\equiv \begin{cases}
    \alpha_3 p g_{A, t(B)} &\text{ if } \abs{A\cap B} = 0 \\
    0&\text{ otherwise.}
  \end{cases}
%  (\tL_{1, 1})_{A, B} &\equiv \alpha_3 p g_{A, h(B)} \\
%  (\tL_{1, 2})_{A, B} &\equiv \alpha_3 p g_{A, t(B)}.
\end{align}

It thus follows that:
\begin{align} \label{eq:M12expansion}
  \tN_{12} - \E\{\tN_{12}\} &= L_{1, 1} + L_{1, 2} + L_{2, 2}.
\end{align}

We first prove two Lemmas on the spectral
properties of the matrices $L_{\eta, \nu}$

\begin{lemma}           
  \label{lem:normbndL21}
  With probability at least $1 - n^{-5}$, we have that
  \begin{align}
    \norm{L_{2, 1}}_2 &\ls \alpha_3 \barn.
  \end{align}
\end{lemma}
\begin{proof}
  Note that:
  \begin{align}
    \Tr\left\{ (L_{2, 1}L_{2, 1}^\sT) \right\} &= \sum_{A_1 \dots A_{m+1}, B_1
  \dots B_m} \prod_{l=1}^r g_{A_l h(B_l)} g_{A_l
  t(B_l)}g_{A_{l+1} h(B_l)}g_{A_{l+1} t(B_l)}
  \end{align}
  Equivalently, letting $B(m)$ be a bridge of length $2m$ we have:
  \begin{align}
    \Tr\left\{ (L_{2, 1}L_{2, 1}^\sT)^m \right\} &= \sum_{\ell\in \frakL(B)}
   \prod_{e=\{u, v\}\in B} g_{\ell(u)\ell(v)}.
  \end{align}
  By Lemma \ref{lem:gentrbnd} it suffices to show that $v_*(B(m)) \le 2m+1$. 
  This argument is already covered in Lemma \ref{lem:normbndeta2} and the
  claim hence follows. 
\end{proof}

\begin{lemma}
    \label{lem:normbndL11}
    With probability exceeding $1 - 2n^{-5}$ the following holds:
    \begin{align}
      \max_{\nu =1, 2}\norm{L_{1, \nu}}_2 &\ls \alpha_3 \barn.
    \end{align}
\end{lemma}
\begin{proof}
  We prove the claim for $L_{1, 1}$. The same argument applies
  for $L_{1, 2}$ with minor modifications. 
  \begin{align}
    \Tr\left\{ (L_{1, 2} L_{1, 2}^\sT)^m \right\} &= \sum_{A_1 \dots
      A_{m+1}, B_1 \dots B_{m}} (\alpha_3 p)^{2m} \prod_{l=1}^{m}g_{A_l
      h(B_\ell)}g_{A_{l+1}h(B_l)}\, .
    \end{align}
    The above a sum over labelings of a bridge $B(m)$ of type 1 and
    class 1, of length $2m$. This is union of a cycle $D(m)$ of length
    $2m$, and $m$ isolated vertices. The lemma follows from Lemma
    \ref{lem:gentrbnd} if $v_*(B(m))\le 2m+1$. But by the above decomposition
    $v_*(B(m))\le v_*(D(m)) + m  = m+1 +m =2m+1$, as in Proposition
    \ref{prop:deviationM11}. This completes the proof.
\end{proof}

We can now prove Proposition \ref{prop:deviationM12}.
\begin{proof}[Proof of Proposition \ref{prop:deviationM12}]
  The intersection of favorable events of lemmas \ref{lem:normbndL21}, 
  \ref{lem:normbndL11} probability at least $1- 5n^{(\Gamma-4)/2}$. 
  The required claim then follows from Lemmas \ref{lem:normbndL21}, \ref{lem:normbndL11}  and triangle
  inequality.
\end{proof}

\subsection{Proof of Proposition \ref{prop:psddecomp}}
\label{subsec:proofpsddecomp}

The intersection of high probability favorable events
of Propositions \ref{prop:deviationM11}, \ref{prop:deviationM22}
and \ref{prop:deviationM12} holds with probability at least
$1-30n^{-5} \ge 1-n^{-4}$ for large enough $n$. 
By Proposition \ref{prop:deviationM11} we already have the required
bounds on $\tN_{11}$ and $\tN_{11}^{-1}$, cf. Eqs.~(\ref{eq:11PSD}) and (\ref{eq:11Bound}). It remains to show that
on the same event:
\begin{align}
  \tN_{22} &\mge \frac{2}{\alpha_1}\tN_{12}^\sT\qproj_n\tN_{12} +
  \frac{1}{n(\alpha_2p - \alpha_1^2)} \tN_{12}^\sT \qproj_n^\perp
  \tN_{12} \, ,
\end{align}
or, equivalently,
\begin{align}
  \text{ Or } \quad \E\{\tN_{22}\} &\mge \E\{\tN_{22}\} - \tN_{22} +  \frac{2}{\alpha_1}\tN_{12}^\sT\qproj_n\tN_{12} +
  \frac{1}{n(\alpha_2 p - \alpha_1^2)} \tN_{12}^\sT \qproj_n^\perp \tN_{12}.
  \label{eq:psdcondition}
\end{align}

Let $\barW, W \in \reals^{3\times 3}$ be two matrices that satisfy,
for $a,b\in \{0,1,2\}$:
\begin{align}
  \barW_{ab} &= \norm{\proj_a \E\{\tN_{22}\}\proj_b}_2 \\
  W_{ab} &\ge \norm{\proj_a (\tN_{22} - \E\{\tN_{22}\}) \proj_b }_2 +
  \frac{2 }{\alpha_1} \norm{\qproj_n^\perp \tN_{12}\proj_a}_2
  \norm{\qproj_n^\perp \tN_{12}\proj_b}_2 \nonumber \\&\quad+
  \frac{1}{n(\alpha_2 p - \alpha_1^2)}
  \norm{\qproj_n \tN_{12}\proj_a}_2\norm{\qproj_n \tN_{12}\proj_b}_2.
\end{align}
By expanding the Rayleigh quotient of each term in
\myeqref{eq:psdcondition},  and noting that $\barW_{ab}=0$ for $a\neq b$,
it is straightforward to see that \myeqref{eq:psdcondition} holds if
\begin{align}
  \alpha_2 p - \alpha_1^2 &\ge 0\, ,\\
  \barW&\mge W\, .
\end{align}
The first condition correspond to assumption
(\ref{eq:Assumption2BigPropo}).
For the second one, we develop explicit expressions of $\barW$, $W$ as
follows.
For $\barW$, we use Proposition \ref{prop:EM22eig}, that yields
immediately $\barW_{a,b} = 0$ for $a\neq b$ as claimed, and $\barW_{0,0}$,
$\barW_{1,1}$, $\barW_{2,2}$ as in Eqs.~(\ref{eq:barW00}),
(\ref{eq:barW11}),  (\ref{eq:barW22}).

In order to develop expressions for  $W$ we note that it is sufficient
to guarantee
\begin{align}
  W_{ab} &\ge \norm{\proj_a (\tN_{22}- \E\{\tN_{22}\})\proj_b}_2 \nonumber \\
  &\quad+ \frac{2}{\alpha_1} \left( \norm{\qproj_n^\perp\E\{\tN_{12}\}\proj_a}_2
  + \norm{\tN_{12} - \E\{\tN_{12}\}}_2\right) 
\left( \norm{\qproj_n^\perp\E\{\tN_{12}\}\proj_b}_2
+ \norm{\tN_{12} - \E\{\tN_{12}\}}_2\right) \nonumber \\
  &\quad+ \frac{1}{n(\alpha_2 p - \alpha_1^2)} \left(
  \norm{\qproj_n\E\{\tN_{12}\}\proj_a}_2
  + \norm{\tN_{12} - \E\{\tN_{12}\}}_2\right) 
\left( \norm{\qproj_n\E\{\tN_{12}\}\proj_b}_2
+ \norm{\tN_{12} - \E\{\tN_{12}\}}_2\right).
\end{align}
Using  the upper bounds 
in Propositions \ref{prop:EM12eig}, \ref{prop:deviationM22}, \ref{prop:deviationM12}
we obtain the expressions in Eqs.~(\ref{eq:W00}) to (\ref{eq:W22}). This
completes the proof.

\bibliographystyle{amsalpha}
\bibliography{all-bibliography}

\end{document}